\documentclass[]{article}

\usepackage{latexsym}
\usepackage{amsmath,amssymb,amsthm}
\usepackage{enumerate}
\usepackage{tikz}
\usepackage{marvosym}
\colorlet{dblue}{blue!40!black}
\usepackage[colorlinks,linkcolor=dblue,citecolor=dblue,urlcolor=dblue]{hyperref}

\usetikzlibrary{arrows,automata,backgrounds,calc,chains,
  decorations.pathmorphing,decorations.fractals,decorations.markings,
  decorations.pathreplacing,decorations.footprints,decorations.text,
  fadings,fit,folding,patterns,positioning,mindmap,shadows,
  shapes.geometric,shapes.symbols,through,trees,plotmarks,
  shapes.multipart}

\theoremstyle{plain}
\newtheorem{theorem}{Theorem}
\newtheorem{corollary}[theorem]{Corollary}
\newtheorem{lemma}[theorem]{Lemma}
\newtheorem{proposition}[theorem]{Proposition}

\newtheorem*{maintheorem*}{Main Theorem}
\theoremstyle{definition}
\newtheorem{definition}[theorem]{Definition}
\newtheorem{remark}[theorem]{Remark}

\definecolor{mediumblue}{RGB}{0,0,205}

\renewcommand{\labelenumi}{\theenumi}

\newcommand{\tup}[1]{\langle#1\rangle}
\newcommand{\pair}[2]{\tup{#1,#2}}

\newcommand{\nb}{\nobreakdash}

\newcommand{\msf}{\mathsf}
\newcommand{\mbb}{\mathbb}

\newcommand{\mcl}{\mathcal}
\newcommand{\mrm}{\mathrm}

\newcommand{\nat}{\mbb{N}}

\newcommand{\myall}[2]{\forall #1.\,#2}

\renewcommand{\emptyset}{\varnothing}
\newcommand{\emptyword}{\varepsilon}

\newcommand{\length}[1]{|#1|}
\newcommand{\card}[1]{|#1|}

\newcommand{\fap}[2]{#1(#2)}

\newcommand{\powerset}{\fap{\wp}}

\newcommand{\funin}{\mathrel{:}}

\newcommand{\where}{\mid}

\tikzset{math/.style = {execute at begin node=$, execute at end node=$}}
\tikzset{label/.style = {execute at begin node=\scriptstyle, execute at end node=}}
\tikzset{node/.style={rectangle,rounded corners=1mm,draw=black,inner sep=1mm,outer sep=1mm}}
\newcommand{\tvs}{\_}

\newcommand{\computablewrt}[1]{${#1}$\nobreakdash-com\-pu\-table}

\newcommand{\rcomputable}{\computablewrt{r}}
\newcommand{\computabilitywrt}[1]{${#1}$\nobreakdash-com\-pu\-ta\-bi\-lity}

\newcommand{\rcomputability}{\computabilitywrt{r}}

\newcommand{\recognizablewrt}[1]{${#1}$\nobreakdash-rec\-og\-niz\-able}

\newcommand{\rrecognizable}{\recognizablewrt{r}}

\newcommand{\lang}[1]{\mrm{Lang}(#1)}

\newcommand{\str}[1]{#1^{\nat}}

\newcommand{\descsetexpmid}{\mathrel{\vert}}

\newcommand{\descsetexpBigmid}{\mathrel{\Big\vert}}

\newcommand{\descsetexp}[2]{\left\{{#1}\descsetexpmid{#2}\right\}}

\newcommand{\descsetexpBig}[2]{\Bigl\{{#1}\descsetexpBigmid{#2}\Bigr\}}

\newcommand{\setexp}[1]{\left\{{#1}\right\}}

\newcommand{\mocDFAsover}{\fap{\msf{DFA}}}

\newcommand{\mocDFAswithidover}{\fap{\msf{DFA}_{\mrm{id}}}}
\newcommand{\mocFSTsover}{\fap{\msf{FST}}}
\newcommand{\moctwowayFSTsover}{\fap{\text{2-$\msf{FST}$}}}
\newcommand{\mocTMsover}{\fap{\msf{TM}}}
\newcommand{\mocTMDsover}{\fap{\msf{TMD}}}

\begin{document}

\title{Regularity Preserving but not Reflecting Encodings}

\author{%
  J\"{o}rg Endrullis \quad Clemens Grabmayer \quad Dimitri Hendriks \\[2ex]
  \small{Department of Computer Science} \\
  \small{VU University Amsterdam}
}

\date{}

\maketitle

\thispagestyle{plain}
\pagestyle{plain}

\begin{abstract}
  Encodings, that is, injective functions from words to words, 
  have been studied extensively in several settings.
  In computability theory 
  the notion of encoding is crucial
  for defining computability on arbitrary domains,
  as well as for comparing the power of models of computation.
  In language theory much attention has been devoted 
  to regularity preserving functions.
  
  A natural question arising in these contexts is:
  Is there a bijective encoding such that its image function preserves regularity of languages,
  but its preimage function does not?
  Our main result answers this question in the affirmative:
  For every countable class~$\mcl{L}$ of languages
  there exists a bijective encoding~$f$ such that 
  for every language $L \in \mcl{L}$ 
  its image~$f[L]$ is regular.
  
  Our construction of such encodings
  has several noteworthy consequences.
  Firstly, anomalies arise when models of computation
  are compared with respect to a known concept of implementation 
  that is based on encodings which are not required to be computable:
  Every countable decision model 
  can be implemented, in this sense, by finite-state automata, even via bijective encodings.
  Hence deterministic finite-state automata 
  would be equally powerful as Turing-machine deciders. 
  
  A second consequence concerns the recognizability of sets of natural numbers
  via number representations and finite automata.
  A set of numbers is said to be recognizable with respect to a representation 
  if an automaton accepts the language of representations.
  Our result entails 
  that there exists a number representation 
  with respect to which every recursive set is recognizable.
\end{abstract}

\section{Introduction}\label{sec:intro}
In order to define computability of number-theoretic functions
through computational models that operate on strings of symbols from an alphabet $\Sigma$
(rather than defining computability directly via recursion schemes)
one usually employs \emph{(number) representations}, that is, 
injective functions $r : \nat \to \Sigma^*$.
A function $f : \nat \to \nat$ is called \emph{\rcomputable} 
  (\emph{computable by a Turing machine using representation~$r$})
if there exists a Turing-computable function $\varphi : \Sigma^* \to \Sigma^*$
such that $\varphi \circ r = r \circ f$.
For representations~$r$ that are informally computable 
(i.e.,\ there is a machine-implementable algorithm that always terminates, and computes~$r$),
it can be argued on the basis of Church's thesis
(similar as e.g.\ in~\cite[p.~28]{roge:1967})
that \rcomputability\ does not depend on the specific choice of $r$,
and coincides with partial recursiveness.

Shapiro~\cite{shap:1982} studied the influence that (unrestricted) bijective representations~$r$ 
have on the notion of \rcomputability.
He found that the only functions that are \rcomputable\ with respect to all bijective representations~$r$
are the almost constant and almost identity functions;
and that there are functions that are not \rcomputable\ for any representation~$r$.
Furthermore, he defines `acceptable' number representations: 
a bijective representation~$r$ is called `acceptable' 
if the successor function lifted to the $r$-coded natural numbers is Turing computable. 
He goes on to show that, a representation~$r$ is acceptable,
if and only if \rcomputability\ coincides with partial recursiveness. 

In this paper we focus on the notion of computability by \emph{finite automata} of \emph{sets} of natural numbers. 
In particular, we investigate how number representations determine the sets of natural numbers 
that are computable by finite-state automata. 
Such sets are called `recognizable':
a set $S \subseteq \nat$ 
is called \emph{\rrecognizable} (\emph{recognizable with respect to representation~$r$}), 
if there is a finite automaton that for all $n \in \nat$ decides
membership of $n$ in $S$ when $r(n)$ is given to it as input.

We are interested in comparing representations~$r$ 
with respect to their computational power as embodied by the \rrecognizable\ sets.
This idea gives rise to a hierarchy via a subsumption preorder between representations:
$r_1 : \nat \to A^*$
subsumes $r_2 : \nat \to B^*$ 
if all \recognizablewrt{r_2} sets are also \recognizablewrt{r_1}. 
There are several natural questions concerning this preorder;
to name a few:
\begin{enumerate}
  \item When does a number representation subsume another?
  \item Is the hierarchy proper: do there exist representations $r_1$ and $r_2$
    such that $r_1$ subsumes $r_2$, but not vice versa?  
  \item Is there a representation that subsumes all others?
  \item Is every (injective) number representation subsumed by a bijective number representation? 
  \item What classes $\mathcal{C} \subseteq \powerset{\nat}$ of sets of natural numbers
     are recognizable with respect to a number representation?
\end{enumerate}

As our computational devices are finite automata,
all of these questions boil down to problems in language theory.
In particular the comparison of number representations is
intimately connected with \emph{encodings}, injective mappings from words to words,
that have the property that their image function preserves regularity of languages.
For bijective number representations $f : \nat \to A^*$ and $g : \nat \to B^*$,
we have that $f$ subsumes $g$ if and only if the set function
\begin{align*}
  (f \circ g^{-1})[\_]
\end{align*}
preserves regularity of languages;
here we use the notation $h[\_]$ to denote the image function of a function $h$. 
Regularity preserving functions 
play an important role in different areas of computer science,
and have been studied extensively.
An important result in this area is the work~\cite{pin:silv:2005,pin:2009} of Pin and Silva, 
providing a characterization of regularity preservation of preimage functions
in terms of uniformly continuous maps on the \emph{profinite topology}.

A natural question that presents itself then is the following:%
\smallskip
\begin{quote}
  \textit{Are there bijective functions $f : \Sigma^* \to \Sigma^*$ whose
  image function $f[\_]$ preserves regularity of languages, 
  but whose preimage function $f^{-1}[\_]$ does not?}
\end{quote}
\smallskip
For \emph{bijective} word functions we experienced this to be a very challenging question,
which to the best of our knowledge, has remained unanswered in the literature.
Using the results of~\cite{pin:silv:2005,pin:2009}, it can equivalently be formulated as follows:
\smallskip
\begin{quote}
  \textit{Are there bijective functions $f : \Sigma^* \to \Sigma^*$
  such that $f$ is uniformly continuous,
  but $f^{-1}$ is not uniformly continuous in the profinite topology?}
\end{quote}
\smallskip
Concerning recognizable sets and the hierarchy of number representations,
the question translates to:
\smallskip
\begin{quote}
  \textit{Are there bijective number representations $f$ and $g$ such that $f$ strictly subsumes $g$?}
\end{quote}
\smallskip
If this were not the case, subsumption would imply equivalence for bijective number representations,
and the hierarchy would collapse.

Our main result (Theorem~\ref{thm:main}), which allows us to answer all of the above questions, is the following:
\begin{maintheorem*}
  For every countable class~$\mcl{L} \subseteq \powerset{\Sigma^*}$ of languages 
  over a finite alphabet $\Sigma$, and for every alphabet $\Gamma$ with $\card{\Gamma} \ge 2$,
  there exists a bijective encoding~$f : \Sigma^* \to \Gamma^*$ such that 
  for every language $L \in \mcl{L}$ 
  its image~$f[L]$ is regular.
\end{maintheorem*}

With respect to computability theory and recognizable sets of natural numbers, 
this result can be restated as follows:
\smallskip
\begin{quote}
  \textit{%
    For every countable decision model $\mcl{M} \subseteq \powerset{\nat}$,
    there exists a bijective representation $f : \nat \to \Sigma^*$ such that 
    every set $M \in \mcl{M}$ is $f$\nb-recognizable.} 
\end{quote}
\smallskip
As a direct consequence, when allowing for arbitrary bijective number representations,
we find the unsought: 
\begin{align}
  \begin{gathered}
    \text{\textit{Finite automata are as strong as Turing-machine deciders.}}
  \end{gathered}
  \tag{\Lightning}
  \label{anomaly}
\end{align}
That is, there is a bijective representation such that 
finite automata can recognize any computable set of natural numbers.

Our result also has consequences in the context of the work 
by Boker and Dershowitz on comparing the power of computational models, as described below.
Models over different domains are typically compared
with the help of encodings that translate between different number representations.
In order to prevent encodings from changing the nature of the problem,
they are usually required to be `informally algorithmic', `informally computable', or `effective'
(see e.g.\ \cite[p.\hspace*{1pt}27]{roge:1967}).
However, the latter concepts are rather vague, and in any case non-mathematical. 
Therefore they are unsatisfactory from the viewpoint of a rigorous conceptual analysis. 

In the formal approach for comparing models of computation 
proposed by Boker and Dershowitz in~\cite{boke:2004,boke:ders:2006,boke:2008},
encodings are merely required to be injective.
On the basis of this stipulation,
a computational model $\mcl{M}_2$ is defined to be `at least as powerful as' $\mcl{M}_1$, 
denoted by $\mcl{M}_1 \;\lesssim\; \mcl{M}_2$,
if there exists an encoding $\rho : \Sigma_1^* \to \Sigma_2^*$ 
with the property that
for every function $f$ computed by $\mcl{M}_1$ 
there is a function $g$ computed by $\mcl{M}_2$
such that the following diagram commutes:
\begin{align*}
  \begin{aligned}
    \begin{tikzpicture}[baseline=2ex,node distance=25mm]
      \node (l) {$\Sigma_1^*$};
      \node (r) [right of=l] {$\Sigma_2^*$};
      \node (l') [below of=l,node distance=15mm] {$\Sigma_1^*$};
      \node (r') [right of=l'] {$\Sigma_2^*$};
      \draw[->] (l) -- node [above] {$\rho$} (r);
      \draw[->] (l') -- node [below] {$\rho$} (r');
      \draw[->] (l) -- node [left] {$f\in\mcl{M}_1$} (l');
      \draw[->] (r) -- node [right] {$g\in\mcl{M}_2$} (r');
    \end{tikzpicture}
  \end{aligned}
  && \hspace*{2ex}
  \begin{aligned}
    \mcl{M}_1 \;\lesssim_{\rho}\; \mcl{M}_2
  \end{aligned}
\end{align*}
(In order to highlight the encoding used, \mbox{$\mcl{M}_1 \;\lesssim_{\rho}\; \mcl{M}_2$} is written.)
Although encodings are not required to be (informally) computable, this approach
works quite well in practice.  
                                          
However, in analogy to what we found for recognizability, 
one runs into the anomaly~\eqref{anomaly} again,
namely when comparing the power of decision models with the preorder $\lesssim$.  
Our main result implies
$\mcl{C} \;\lesssim\; \msf{DFA}$
for every countable class of decision problems~$\mcl{C}$,
where $\msf{DFA}$ denotes the class of deterministic finite-state automata.
Even stronger, it follows that  there is a bijective encoding $\rho$ such that 
$\mcl{C} \;\lesssim_{\rho}\; \msf{DFA}$.
As a consequence we obtain that 
$\msf{TMD} \lesssim_{\rho} \msf{DFA}$ holds
for the class $\msf{TMD}$ of Turing-machine deciders,
and a bijective encoding~$\rho$.

\subsection*{Further Related Work}

For a general introduction to automata and regular languages we refer to~\cite{saka:2009,hopc:motw:ullm:2000}.
We briefly mention related work with respect to regularity preserving functions
apart from work~\cite{pin:silv:2005,pin:2009} of Pin and Silva that we have already discussed above.
The works \cite{stea:hart:1963,kosa:1970,mato,seif:1974,kosa:1974,seif:mcna:1976,kvet:koub:1999}
investigate regularity preserving relations;
in particular,
\cite{seif:mcna:1976} provides a characterization of prefix-removals that preserve regularity.
In \cite{pin:saka:1982}, Pin and Sakarovitch study operations and transductions that preserve regularity.
In \cite{koze:1996}, Kozen gives a characterization of word functions over a one-letter alphabet 
whose preimage function preserves regularity of languages.
The paper~\cite{bers:boas:cart:peta:pin:2006} by Berstel, Boasson, Carton, Petazzoni and Pin
characterizes language preserving `filters'; a filter is a set $F \subseteq \nat$
used to delete letters from words of the language as indexed by elements of the filter.

\section{Preliminaries}\label{sec:prelims}
We use standard terminology and notation, see, e.g., \cite{allo:shal:2003}.
Let $\Sigma$ be an alphabet, 
i.e., a finite non-empty set of symbols.
We denote by $\Sigma^*$ the set of all finite words over~$\Sigma$,
and by $\emptyword$ the empty word.
%
The set of infinite sequences over~$\Sigma$ is 
$\str{\Sigma} = \{ \sigma \where \sigma : \nat \to \Sigma \}$
with $\nat = \{0,1,2,\ldots\}$, the set of natural numbers. 

A deterministic finite-state automaton (DFA)
is a tuple $A = \tup{Q,\Sigma,\delta,q_0,F}$
consisting of 
a finite set of states $Q$,
an input alphabet $\Sigma$, 
a transition function~$\delta : Q \times \Sigma \to Q$,
an initial state $q_0 \in Q$, and
a set $F \subseteq Q$ of accepting states.
The transition function~$\delta$ is extended to $\delta^* : Q \times \Sigma^* \to Q$ by 
\begin{align*}
  \delta^*(q,\emptyword) = q && 
  \delta^*(q,aw) = \delta^*(\delta(q,a),w) \,,
\end{align*}
for all states $q \in Q$, letters $a \in \Sigma$ and words $w \in \Sigma^*$.
We will write just $\delta$ for $\delta^*$.
A word $w \in \Sigma^*$ is accepted by $A$ if $\delta(q_0,w)$, 
the state reached after reading $w$, is an accepting state.
We write $\lang{A}$ for the language accepted by the automaton~$A$,
i.e., $\lang{A} = \{ w \in \Sigma^* \where \delta(q_0,w) \in F \}$.

A DFA with output (DFAO) is a tuple $\tup{Q,\Sigma,\delta,q_0,\Delta,\lambda}$,
with the first four components as in the definition of a DFA, 
but with, instead of a set of accepting states, an output alphabet~$\Delta$, 
and an output function $\lambda : Q \to \Delta$.
A DFAO $A = \tup{Q,\Sigma,\delta,q_0,\Delta,\lambda}$ realizes a function mapping words over $\Sigma$ 
to letters in $\Delta$;
we denote this function also by $A$, that is, we define $A : \Sigma^* \to \Delta$ by
\begin{align*}
  A(w) = \lambda(\delta(q_0,w)) && \text{for all $w \in \Sigma^*$.}
\end{align*}

A DFA $\tup{Q,\Sigma,\delta,q_0,F}$
can thus be viewed as a DFAO $\tup{Q,\Sigma,\delta,q_0,\{0,1\},\chi_F}$
where $\chi_F$ is the characteristic function of $F$;
instead of a state being accepting or not,
it has output $1$ or $0$ respectively.

Two automata $A$ and $A'$ over the same input alphabet $\Sigma$ 
are equivalent on a set $X \subseteq \Sigma^*$ 
if $A(x) = A'(x)$ for all \mbox{$x \in X$}.
We use this functional notation also for a DFA $A$,
stipulating: $A(x) = 1$ ($A(x) = 0$) iff $A$ accepts $x$ ($A$ does not accept $x$).  

A partition $P$ of a set $U$ is a family of sets
$P \subseteq \powerset{U}$ 
such that $\emptyset \not \in P$,
$\bigcup_{A \in P} A = U$,
and for all $A,B \in P$ with $A \neq B$,
$A \cap B = \emptyset$.

The pigeon hole principle (PHP) states that
if $n$ pigeons are put into $m$ pigeonholes 
with $n > m$, then at least one pigeonhole contains more than one pigeon.
PHP for infinite sets is that if infinitely many pigeons 
are put into finitely many holes, then one hole must contain infinitely many pigeons.

Let $A$, $B$, $X$, and $Y$ be sets, with $X \subseteq A$ and $Y \subseteq B$.
For a function $A \to B$ we write $f[X]$ for the image of $X$ under $f$,
that is, $f[X] = \{ f(x) \where x \in X\}$.
Likewise, we write $f^{-1}[Y]$ for the preimage of $Y$ under $f$, 
i.e., $f^{-1}[Y] = \{ x \where f(x) \in Y \}$.

A function $F : \powerset{\Sigma^*} \to \powerset{\Gamma^*}$ \emph{preserves regularity}
if $F(L)$ is regular whenever $L$ is a regular language.

\section{Main Results}\label{sec:main}

In this section, we prove our main results.
We first work towards Theorem~\ref{thm:main} stating that 
for every countable class~$\mcl{L}$ of languages 
there exists a bijective encoding~$f : \Sigma^* \to \Gamma^*$ 
such that $f[L]$ is regular for every language $L \in \mcl{L}$.
The proof proceeds in two stages:
we first prove the existence of injective 
encodings (Lemma~\ref{lem:inj:reg:crea}),
and then strengthen this result to bijective functions (Lemma~\ref{lem:inj2bij}).
From Theorem~\ref{thm:main} we then obtain the existence of bijective functions that are regularity preserving  
but not regularity reflecting, Corollary~\ref{cor:puzzle}.

For injective encodings~$f$ we cannot require the images $f[L]$ to be regular;
instead we require $f[L]$ to be recognizable among $f[\Sigma^*]$.
This leads to the notion of `relatively~regular~in'.

\begin{definition}\label{def:rel:reg}
  Let $L,M \subseteq \Sigma^*$ be formal languages over the alphabet~$\Sigma$, with $L \subseteq M$.
  Then \emph{$L$ is relatively regular in $M$}
  if there exists a regular language $R$ such that $L = R \cap M$.
\end{definition}
So a regular language $S \subseteq \Sigma^*$ 
is relatively regular in $\Sigma^*$.

\begin{lemma}\label{lem:inj:reg:crea}
  Let $\Sigma$ and $\Gamma$ be finite alphabets, with $\card{\Gamma} \ge 2$.
  Let $\mcl{L} \subseteq \powerset{\Sigma^*}$ be a countable set 
  of formal languages. 
  There exists an injective function $f : \Sigma^* \to \Gamma^*$ 
  such that for every $L \in \mcl{L}$,
  $f[L]$ is relatively regular in 
  $f[\Sigma^*]$.
\end{lemma}
\begin{proof}
  Let $L_1, L_2, L_3, \ldots$ be an enumeration of $\mcl{L}$, 
  and let $v_0, v_1, v_2, \ldots$ be an enumeration of $\Sigma^*$.
  For $i \ge 1$, we write $\chi_i : \Sigma^* \to \{0,1\}$ 
  for the characteristic function of $L_i$, that is,
  \begin{align*}
    \chi_i(v) = 
    \begin{cases}
      1 & \text{if $v \in L_i$} \\
      0 & \text{otherwise}
    \end{cases}
    && \text{for all $v \in \Sigma^*$.}
  \end{align*}
  Without loss of generality we assume
  $\Gamma = \{0,1,\ldots,k-1\}$ for some $k \ge 2$.
  Define the function $f : \Sigma^* \to \Gamma^*$ by
  \begin{align*}
    f(v_n) = \chi_1(v_n) \, \chi_2(v_n) \, \cdots \, \chi_n(v_n) \,, && \text{for all $n \in \nat$.}
  \end{align*}
  For every $i = 1,2,\ldots$, we construct a DFAO~$A_i$
  and show that $f[L_i] = \lang{A_i} \cap f[\Sigma^*]$ 
  witnessing that $f[L_i]$ is relatively regular in $f[\Sigma^*]$, 
  as required.  

  Fix an arbitrary integer $i \ge 1$.
  Define
  $A_i = \tup{Q,\Gamma,\delta,0,\lambda}$ where 
  $Q = \{0,1,\ldots,i-1\} \cup \{i_0,i_1\}$,
  $\Gamma = \{0,1,\ldots,k-1\}$,
  $\delta : Q \times \Gamma \to Q$ is defined by
  \begin{align*}
    \delta(q,a) & = q+1 && \text{for all $q \in \{0,1,\ldots,i-2\}$ and $a \in \Gamma$,} \\
    \delta(i-1,0) & = i_0 \\
    \delta(i-1,a) & = i_1 && \text{for all $a \in \Gamma \setminus \{0\}$,} \\
    \delta(i_j,a) & = i_j && \text{for all $j \in \{0,1\}$,}
  \end{align*}
  and 
  $\lambda : Q \to \Gamma$ is defined, for all $q \in \{0,1,\ldots,i-1\}$, by
  \begin{align*}
    \lambda(q) = \chi_i(v_q) &&
    \lambda(i_0) = 0  &&
    \lambda(i_1) = 1 \,.
  \end{align*}
  The automaton~$A_i$ is depicted in Figure~\ref{fig:auto}.
    \begin{figure*}[tbh]
  \begin{center}
    \begin{tikzpicture}[math,node distance=25mm]
      \node [node] (0) {0 / \chi_i(v_0)};
      \node [node,right of=0] (1) {1 / \chi_i(v_1)};
      \node [node,right of=1] (2) {2 / \chi_i(v_2)};
      \node [right of=2,node distance=20mm] (3) {\cdots};
      \node [node,right of=3] (i-1) {i-1 / \chi_i(v_{i-1})};
      \node [node,below left of=i-1] (i0) {i_0 / 0};
      \node [node,below right of=i-1] (i1) {i_1 / 1};
      \draw [->] (0) -- node [label,above] {\Gamma} (1);
      \draw [->] (1) -- node [label,above] {\Gamma} (2);
      \draw [->] (2) -- node [label,above] {\Gamma} (3);
      \draw [->] (3) -- node [label,above] {\Gamma} (i-1);
      \draw [->] (i-1) -- node [label,left,inner sep=2mm] {0} (i0);
      \draw [->] (i-1) -- node [label,right,inner sep=2mm] {1,2,\ldots,k-1} (i1);
      \draw [->] (i0) to[out=-60,in=-120,looseness=4] node [label,below] {\Gamma} (i0);
      \draw [->] (i1) to[out=-60,in=-120,looseness=4] node [label,below] {\Gamma} (i1);
    \end{tikzpicture}
  \end{center}
  \caption{The DFAO $A_i$ used in Lemma~\ref{lem:inj:reg:crea}. A transition labelled~$\Gamma$ stands for transitions labelled~$a$ for every $a \in \Gamma$.}
  \label{fig:auto}
  \end{figure*}
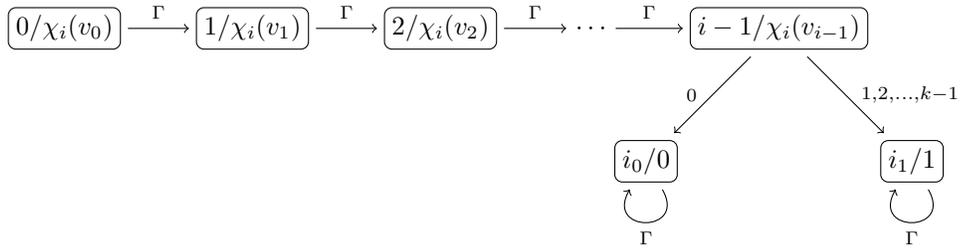

  We show that 
  $f[L_i] = \lang{A_i} \cap f[\Sigma^*]$;
  equivalently,
  for all $v \in \Sigma^*$,
  \begin{align*}
    v \in L_i \iff f(v) \in \lang{A_i} \,.
  \end{align*}
  Let $v \in \Sigma^*$, so $v = v_n$ for some $n \in \nat$.
  Note that $\length{f(v_n)} = n$.
  Hence, if $n \le i$,
  the automaton $A_i$ is in state $n$ after having read the word $f(v_n)$ 
  and outputs $\chi_i(v_n)$.
  If $n \ge i$, then after having read $f(v_n)$, 
  the automaton is in state $i_j$, 
  where $j$ is the $i$-th letter of $f(v_n)$, that is, $j = \chi_i(v_n)$.
  In both cases we get $v_n \in L_i$ if and only if $f(v_n) \in \lang{A_i}$.
\end{proof}

For lifting the result of Lemma~\ref{lem:inj:reg:crea} from injective to bijective encodings,
we need some preliminary notions and results.

\begin{definition}\label{def:attracted}
  Let $U$ be a set and $I,C \subseteq U$.
  The set
  \emph{$C$ is attracted to $I$}
  if $C \subseteq I$
  whenever $C \cap I$ is finite.
  For a partition~$E$ of $U$, we say that \emph{$E$ is attracted to $I$},
  when, for every $C \in E$, $C$ is attracted to $I$.
\end{definition}

Equivalently, $E$ is attracted to $I$ if
for every $C  \in E$, 
if $C \setminus I \ne \emptyset$ then $C \cap I$ is infinite.

\begin{lemma}\label{lem:modify}
  Let $E = \{ C_0, C_1 \ldots, C_{n-1} \}$ be a finite partition of $\Sigma^*$ 
  with $C_i \subseteq \Sigma^*$ a regular set for every $i \in \{0,1,\ldots,n-1\}$. 
  Let $I \subseteq \Sigma^*$ and assume that $E$ is attracted to $I$.
  For every DFA~$A$ 
  there exists a DFA~$A'$ 
  such that $A'$ is equivalent to~$A$~on~$I$
  and the refined partition 
  \begin{align*}
    E' = \{ C'_{i,j} \where i \in \{0,1,\ldots,n-1\},\, j \in \{0,1\} \}
  \end{align*}
  is attracted to~$I$, 
  where,
  for $i \in \{0,1,\ldots,n-1\}$, $j \in \{0,1\}$\,,
  \begin{align*}
    C'_{i,j} = C_i \cap \{ u \in \Sigma^* \where A'(u) = j \} \,.
  \end{align*}
\end{lemma}
\begin{proof}
  We start with $A' = A$,
  and repeatedly adapt $A'$ (and therewith $E'$) until $E'$ is attracted to~$I$,
  in such a way that equivalence with $A$
  is upheld. 
  
  Assume that $E'$ is not attracted to $I$. 
  Then there exist $i \in \{0,1,\ldots,n-1\}$ and $j \in \{0,1\}$ such that 
  $C'_{i,j} \setminus I \ne \emptyset$ but $C'_{i,j} \cap I$ is finite.
  Without loss of generality, assume that $j = 0$.
  Since $C_i = C'_{i,j} \cup C'_{i,1-j}$ 
  it follows that $C_{i} \setminus I \ne \emptyset$ and hence $C_{i} \cap I$ is infinite by assumption. 
  By the pigeonhole principle (for infinite sets) it follows that $C'_{i,1}\cap I$ is infinite.
  Since $C_i$ is a regular set and $A'$ is a finite automaton,
  it follows that $C'_{i,0}$ is regular as it is the intersection of two regular sets.
  As $C'_{i,0} \cap I$ is finite, also $C'_{i,0} \setminus I$ is regular.
  As a consequence we can change the finite automaton~$A'$ to accept the terms in $C'_{i,0} \setminus I$
  (and otherwise to behave as before). 
  This adaptation preserves equivalence,
  and we now have 
  that $C'_{i,j}$ is attracted to~$I$ for $j = 0,1$,
  since, after the adaptation, 
  $C'_{i,0} \cap I = C'_{i,0}$ and $C'_{i,1}\cap I$ is infinite.
  We repeat the procedure until $C'_{i,j}$ is attracted to $I$ 
  for every $i \in \{0,1,\ldots,n-1\}$ and $j \in \{0,1\}$\,.
\end{proof}

%
%

The following lemma, Lemma~\ref{lem:inj2bij}, is a key contribution of our paper.
It states that every injection $f : A \to \Gamma^*$ (with $A$ some countably infinite set)
can be transformed into a bijection $g : A \to \Gamma^*$ 
such that, for all $L \subseteq A$, $g[L]$ is a regular language
whenever $f[L]$ is relatively regular in the image~$f[A]$.
Before proving the lemma, we sketch the construction.
We construct $g$ as the limit of a sequence of adaptations of $f$.
Roughly speaking, we make $f$ `more bijective' in every step.
We let 
\begin{itemize}
  \item $v_0,v_1,v_2,\ldots$ be an enumeration of $A$, and
  \item $w_0,w_1,w_2,\ldots$ be an enumeration of $\Gamma^*$.
\end{itemize}
Figure~\ref{fig:f0} sketches an injective function $f$;
an arrow from $v_i$ to $w_j$ indicates that $f(v_i) = w_j$.

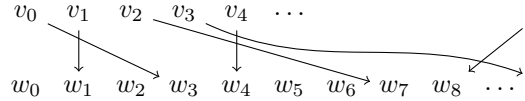
\begin{figure}[h]
  \centering
  \begin{tikzpicture}
    \foreach \i in {0,1,2,3,4} {
      \node (v\i) at (\i*.7cm,0) {$v_\i$};
    }
    \node at (5*.7cm,0) {$\ldots$};
    \foreach \i in {0,1,2,3,4,5,6,7,8} {
      \node (w\i) at (\i*.7cm,-1cm) {$w_\i$};
    }
    \node at (9*.7cm,-1cm) {$\ldots$};
    \begin{scope}[->]
    \draw (v0) -- (w3);
    \draw (v1) -- (w1);
    \draw (v2) -- (w7);
    \draw (v3) to[out=-25,in=160] ($(w8)+(1cm,.2cm)$);
    \draw (v4) -- (w4);
    \draw ($(w8)+(1cm,.8cm)$) -- (w8);
    \end{scope}
  \end{tikzpicture}
  \caption{Starting situation: the function $f$.}
  \label{fig:f0}
\end{figure}

\noindent
The idea is that we change the target of arrows such that all words $w_i$ for $i \in\nat$
become part of the image.
For every natural number $n = 0,1,2,\ldots$ we will pick (while avoiding repetitions) 
a word $v_{k_n}$ (for $k_n \in\nat$) from the input domain
and then adapt $f$ by stipulating $v_{k_n} \mapsto w_n$.
Then, in the limit, every word $w \in \Gamma^*$ will be in the image.
(In order for the limit of the adaptation to be a function, 
we also need to guarantee that every $v_{k_n}$ will be picked precisely once.)

The crucial point of the construction is the following: when changing arrows, we need to ensure that
\begin{enumerate}
  \item [($\star$)] the limit of the process preserves relative regularity.
\end{enumerate}
Note that for bijective functions $g : A \to \Gamma^*$ 
we have that $g[L]$ is relatively regular in the image $g[A] = \Gamma^*$
if and only if $g[L]$ is regular.
Thus if we can ensure ($\star$), 
the resulting bijective function $g$ will have the desired property.

How to pick the $v_{k_n}$ for ensuring ($\star$)?
Let $A_0,A_1,A_2,\ldots$ be an enumeration of all finite automata over the alphabet~$\Gamma$.
We write $u \sim_n v$ if for every $i < n$, the automaton $A_i$
accepts the word $u$ if and only if it accepts the word $v$.
We then pick for every natural number $n\in\nat$, 
a word $v_{k_n}$ such that $f(v_{k_n}) \sim_n w_n$ (and the word $v_{k_n}$ has not been picked before).
In other words, we pick $v_{k_n}$
such that the first $n$ automata $A_0,A_1,\ldots,A_{n-1}$
cannot distinguish $f(v_{k_n})$ from the image $w_n$ after the adaptation $v_{k_n} \mapsto w_n$.
This choice guarantees that every automaton $A_i$ ($i \in \nat$) is only 
affected by a finite number of adaptations, namely the first~$i$ transformation steps.
For every further adaptation ($j > i$), the behavior of the automaton $A_i$ is taken into 
account for the choice of $v_{k_j}$, and as a consequence the modification $v_{k_j} \mapsto w_j$ 
preserves the acceptance behavior of $A_i$.
Then for the limit $g$ of the adaptation process 
we have for almost all $n\in\nat$ that $A_i$ accepts $f(v_n)$ if and only if $A_i$ accepts $g(v_n)$.
In order to guarantee that every $v_i$ will be picked eventually and that the adaptation preserves injectivity, 
we pick among the suitable candidates for $v_{k_n}$ the one which
appears first in the enumeration $w_0,w_1,\ldots$.

\begin{remark}
  There is a caveat here that we will ignore in this sketch of the construction. 
  We actually need to make sure that a word $v_{k_n}$ with these properties exist.
  To ensure this, the equivalence classes with respect to $\sim_n$ must be attracted to the image of $f$.
  This is in general not the case, but we can
  employ Lemma~\ref{lem:modify} to adapt the automata outside of the image of $f$.
  We refer to the proof of Lemma~\ref{lem:inj2bij} for the details.
\end{remark}

We explain this process at the example of the function $f$ given in Figure~\ref{fig:f0}.
For the first step $n=0$, we want to adapt $f$ such that $w_0$ becomes part of the image of $f$.
Note that the relation $\sim_0$ relates all words of $\Gamma^*$. 
As a consequence we can pick any word $v_{k_0}$.
We choose $v_{k_0} = v_1$ since the image $f(v_1)$ appears first in the sequence $w_0,w_1,\ldots$,
and we adapt the function $f$ by $v_1 \mapsto w_0$.

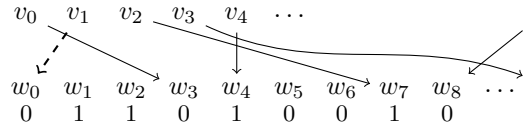
\begin{figure}[h]
  \centering
  \begin{tikzpicture}
    \foreach \i in {0,1,2,3,4} {
      \node (v\i) at (\i*.7cm,0) {$v_\i$};
    }
    \node at (5*.7cm,0) {$\ldots$};
    \foreach \i in {0,1,2,3,4,5,6,7,8} {
      \node (w\i) at (\i*.7cm,-1cm) {$w_\i$};
    }
    \node at (9*.7cm,-1cm) {$\ldots$};
    \foreach \i/\x in {0/0,1/1,2/1,3/0,4/1,5/0,6/0,7/1,8/0} {
      \node at (\i*.7cm,-1.3cm) {{\small $\x$}};
    }
    \begin{scope}[->]
    \draw (v0) -- (w3);
    \draw [thick,dashed] (v1) -- (w0);
    \draw (v2) -- (w7);
    \draw (v3) to[out=-25,in=160] ($(w8)+(1cm,.2cm)$);
    \draw (v4) -- (w4);
    \draw ($(w8)+(1cm,.8cm)$) -- (w8);
    \end{scope}
  \end{tikzpicture}
  \caption{Result of the first adaptation of the function $f$.}
  \label{fig:f1}
\end{figure}

The result of this first adaptation is shown in Figure~\ref{fig:f1}.
For the second step ($n=1$) we want that $w_1$ 
becomes part of the image.
Now $\sim_1$ relates words that have equal behavior with respect to acceptance by the automaton $A_0$.
The numbers $0$ and~$1$ below the words $w_i$ in Figure~\ref{fig:f1}
indicate whether $A_0$ accepts the word $w_i$ (1) or not (0).
The word $w_1$ is accepted by $A_0$ and likewise are $f(v_2)$ and $f(v_4)$.
Among these candidates, we choose $v_{k_1} = v_4$ since $f(v_4)$ 
appears first in the sequence $w_1,w_2,\ldots$.
We modify the function $f$ by setting $v_4 \mapsto w_1$.

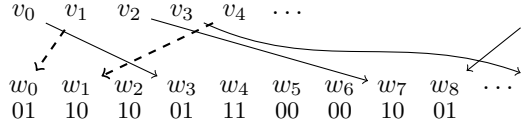
\begin{figure}[h]
  \centering
  \begin{tikzpicture}
    \foreach \i in {0,1,2,3,4} {
      \node (v\i) at (\i*.7cm,0) {$v_\i$};
    }
    \node at (5*.7cm,0) {$\ldots$};
    \foreach \i in {0,1,2,3,4,5,6,7,8} {
      \node (w\i) at (\i*.7cm,-1cm) {$w_\i$};
    }
    \node at (9*.7cm,-1cm) {$\ldots$};
    \foreach \i/\x in {0/01,1/10,2/10,3/01,4/11,5/00,6/00,7/10,8/01} {
      \node at (\i*.7cm,-1.3cm) {{\small $\x$}};
    }
    \begin{scope}[->]
    \draw (v0) -- (w3);
    \draw [thick,dashed] (v1) -- (w0);
    \draw (v2) -- (w7);
    \draw (v3) to[out=-25,in=160] ($(w8)+(1cm,.2cm)$);
    \draw [thick,dashed] (v4) -- (w1);
    \draw ($(w8)+(1cm,.8cm)$) -- (w8);
    \end{scope}
  \end{tikzpicture}
  \caption{Result of the second adaptation of the function $f$.}
  \label{fig:f2}
\end{figure}

The result of the second adaptation is shown in Figure~\ref{fig:f2}.
Now $n=2$ and $\sim_2$ relates words that have equal acceptance behavior
with respect to automata $A_0$ and $A_1$.
To this end, we now write $A_0(w_i)\,A_1(w_i)$ below each word $w_i$ in Figure~\ref{fig:f2}.
The word $w_2$ is accepted by $A_0$ but rejected by $A_1$.
The only (displayed) candidate for $v_{k_2}$ exhibiting the same behavior is $f(v_2)$.
The result of this third adaptation is shown in Figure~\ref{fig:f3}.

\begin{figure}[h]
  \centering
  \begin{tikzpicture}
    \foreach \i in {0,1,2,3,4} {
      \node (v\i) at (\i*.7cm,0) {$v_\i$};
    }
    \node at (5*.7cm,0) {$\ldots$};
    \foreach \i in {0,1,2,3,4,5,6,7,8} {
      \node (w\i) at (\i*.7cm,-1cm) {$w_\i$};
    }
    \node at (9*.7cm,-1cm) {$\ldots$};
    \begin{scope}[->]
    \draw (v0) to[bend left=10] (w3);
    \draw [thick,dashed] (v1) -- (w0);
    \draw [thick,dashed] (v2) -- (w2);
    \draw (v3) to[out=-25,in=160] ($(w8)+(1cm,.2cm)$);
    \draw [thick,dashed] (v4) to[bend right=10] (w1);
    \draw ($(w8)+(1cm,.8cm)$) -- (w8);
    \end{scope}
  \end{tikzpicture}
  \caption{Result of the third adaptation of the function $f$.}
  \label{fig:f3}
\end{figure}
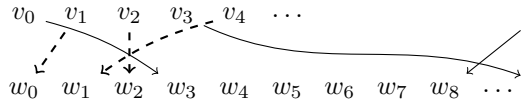

This process continues for every $n \in \nat$, and
the limit of this process is a bijective function $g$ with the desired properties.
The construction is made precise in the proof of Lemma~\ref{lem:inj2bij}.

\begin{lemma}\label{lem:inj2bij}
  Let $A$ be a countably infinite set.
  For every injection $f : A \to \Gamma^*$
  there exists a bijection $g : A \to \Gamma^*$
  such that for all $L \subseteq A$,
  if $f[L]$ is relatively regular in the image $f[A]$,
  then $g[L]$ is a regular language.
\end{lemma}
\begin{proof}
  
  Let $f : A \to \Gamma^*$ be an injective function.
  Let
  \begin{itemize}
    \item $v_0,v_1,v_2,\ldots$ be an enumeration of the set $A$, and let
    \item $w_0,w_1,w_2,\ldots$ be an enumeration of the set $\Gamma^*$, and let
    \item $A_0,A_1,A_2,\ldots$ be an enumeration of all finite-state automata over 
      $\Gamma$.
  \end{itemize}
  We abstract away from the domain by defining $c : \nat \to \Gamma^*$~by
  \begin{align*}
    \text{$c(i) = f(v_i)$ \quad for all $i \in \nat$.}
  \end{align*}   
  In this proof, will construct a bijective function 
  $d : \nat \to \Gamma^*$
  such that 
  for all $S \subseteq \nat$,
  if $c[S]$ is relatively regular in the image $c[\nat]$,
  then $d[S]$ is a regular language.

  Then we can define the function $g : A \to \Gamma^*$ by
  \begin{align*}
    g(w_i) = d(i) && \text{for all $i \in \nat$.}
  \end{align*}
  The bijectivity of $g$ follows immediately from bijectivity of~$d$, 
  and for every $L \subseteq A$ with $f[L]$ is relatively regular in $f[A]$,
  we have:
  let $S \subseteq \nat$ such that $L = \{\, v_i \mid i \in S \,\}$,
  then $c[S] = f(L)$ is relatively regular in $c[\nat] = f[A]$,
  and hence $d[S] = g[L]$ is regular.
  As a consequence, it suffices to construct a function $d$ with the properties above.
  
  For every $n \in \nat$ we will define a finite-state automaton $A'_n$
  with properties described below.
  For each $n \in \nat$, we define the equivalence relation
  ${\sim_n} \subseteq \Gamma^* \times \Gamma^*$
  by
  \begin{align*}
    u \sim_n u' \iff \myall{i \in \{0,1,\ldots,n-1\}}{A'_i(u) = A'_i(u')} \,,
  \end{align*}
  for all $u,u'\in\Gamma^*$. 
  So $\sim_n$ relates words that are not distinguished by the automata $A'_0,A'_1,\ldots,A'_{n-1}$.
  The relation ${\sim_n}$ gives rise to a partition 
  \begin{align*}
    E_n = \{\, [u]_{\sim_n} \where u \in \Gamma^* \,\}
  \end{align*} 
  where $[u]_{\sim_n} = \{ u' \in \Gamma^* \where u \sim_n u' \}$
  is the equivalence class of $u$ with respect to ${\sim_n}$. 

  We are ready to carry out the central construction.
  We first describe what objects we will define and their main properties.
  Then we give specifications of these objects,
  and afterwards we show that the objects are well-defined and that the properties hold. 
  
  We are going to define, by induction on $n \in \nat$,
  \begin{enumerate}
    \item 
      a natural number $k_n \in \overline{K}_{\!n}$\,,
      where
      \begin{align*}
        \overline{K}_{\!n} & = \nat \setminus K_{n} &&& K_n &= \{k_0,k_1,\ldots,k_{n-1}\} \,,
      \end{align*}
      such that
      \begin{align}
        c(k_n) \sim_n w_n \,.
        \label{eq:ckn:sim:wn}  
      \end{align}

    \item a bijective map $d_n : \nat \to I_n$\,, 
      where
      \begin{align*}
        I_n & = c[\overline{K}_{\!n}] \cup W_{\!n} &&&  W_{\!n} & = \{w_0,w_1,\ldots,w_{n-1}\} \,,
      \end{align*}
    
    \item an automaton $A'_n$\,,
      such that 
      \begin{enumerate}
        \item $\myall{u\in I_{n}}{A_n(u) = A'_n(u)}$\,, and
        \item $E_{n+1}$ is attracted to $I_{n+1}$.
      \end{enumerate}

  \end{enumerate}
  \medskip

  We will now give specifications of (i)--(iii), and thereafter show that they are well-defined,
  and satisfy the above mentioned properties.
  \begin{enumerate}
    \item 
      We define $k_n \in \overline{K}_{\!n}$ as follows:
      \begin{align*}
        k_n \;=\; c^{-1}(\min([w_n]_{\sim_n} \cap c[\overline{K}_{\!n}])) \,,
      \end{align*}
      where $\min$ denotes the minimum with respect to the order
      ${<}$ on $\Gamma^*$
      given by the enumeration $w_0 < w_1 < w_2 < \cdots$\,.
    \item 
      The map $d_n : \nat \to I_n$ is defined for all $i \in \nat$ by
      \begin{align*}
        d_n(i) = 
        \begin{cases}
          w_j  & \text{if $i = k_j$ for some $j \in \{0,1,\ldots,n-1\}$\,,} \\
          c(i) & \text{if $i \in \overline{K}_{\!n}$.}
        \end{cases}
      \end{align*}
    \item 
      For $A'_n$ we choose the automaton $A'$ guaranteed to exist by Lemma~\ref{lem:modify}
      where the lemma is invoked with $A = A_n$, $I = I_{n+1}$, and $E = E_n$. 
      
  \end{enumerate}

  We note that ${\sim_0} = \Gamma^* \times \Gamma^*$,
  and, by definition of $\sim_n$, it follows that
  \begin{align}
    \text{$E_n$ is finite and every $C \in E_n$ is a regular language.}
    \label{eq:reg}
  \end{align}

  We prove that (i)--(iii) are well-defined, 
  and that the following properties hold:
  \begin{enumerate}[(a)]
    
    \item \label{it:intersect}
      $c[\overline{K}_{\!n}] \cap W_n = \emptyset$\,,
    \item \label{it:attract}
      $E_n$ is attracted to $I_n$,
    \item \label{it:dn}
      $d_n$ is a bijection, and
    \item \label{it:kn}
      $k_n$ is well-defined.
  \end{enumerate}

  We first show that, for every $n \in \nat$, 
  items~\eqref{it:dn} and~\eqref{it:kn} follow from \eqref{it:intersect} and~\eqref{it:attract}.
  \begin{enumerate}[(a)]
    \setcounter{enumi}{2}
    \item 
      We show that $d_{n} : \nat \to I_n$ is bijective.
      Surjectivity of $d_n$ is immediate by definition of $I_n$.
      To show injectivity, assume there exist $i_1,i_2\in\nat$
      such that $d_{n}(i_1) = d_{n}(i_2)$.
      From \eqref{it:intersect} and the definition of $d_n$, it follows that either 
      $i_1,i_2 \in K_{n}$
      or
      $i_1,i_2 \in \overline{K}_{\!n}$.
      If $i_1,i_2 \in K_{n}$, then there are $j_1,j_2 \in \{0,1,\ldots,n-1\}$
      such that $i_1 = k_{j_1}$ and $i_2 = k_{j_2}$.
      Then it follows that $w_{j_1} = d_{n}(i_1) = d_{n}(i_2) = w_{j_2}$, which can only be in case if $j_1 = j_2$.
      Hence $i_1 = k_{j_1} = k_{j_2} = i_2$.
      If $i_1,i_2 \in \overline{K}_{\!n}$, then $c(i_1) = d_n(i_1) = d_n(i_2) = c(i_2)$.
      By injectivity of $c$ we have $i_1 = i_2$.
      
    \item 
      To see that $k_n$ is well-defined it suffices to show that
      $[w_n]_{\sim_n} \cap c[\overline{K}_{\!n}]$ is non-empty.
      By \eqref{it:attract} we know that 
      either $|[w_n]_{\sim_n} \cap I_{n}| = \infty$, 
      or $[w_n]_{\sim_n} \subseteq I_{n}$:
      \begin{itemize}
        \item For $|[w_n]_{\sim_n} \cap I_{n}| = \infty$, it follows that $|[w_n]_{\sim_n} \cap c[\overline{K}_{\!n}]| = \infty$.
        \item 
          For $[w_n]_{\sim_n} \subseteq I_{n}$, we get $w_n \in ( [w_n]_{\sim_n} \cap I_{n} )$. 
          Since $I_n = c[\overline{K}_{\!n}] \cup W_{\!n}$, and $w_n\notin W_n$,
          \eqref{it:intersect} entails
          $w_n \in c[\overline{K}_{\!n}]$,
          and hence
          $[w_n]_{\sim_n} \cap c[\overline{K}_{\!n}] \ne \emptyset$.
      \end{itemize} 
      We note that, indeed, 
      \eqref{eq:ckn:sim:wn} follows by the choice of $k_n$.

  \end{enumerate}

  We now prove \eqref{it:intersect} and \eqref{it:attract} by induction on $n\in\nat$.
  For the base case we have:
  \begin{enumerate}[(a)]
    \item[\eqref{it:intersect}] 
       $c[\overline{K}_{\!0}] \cap W_0 = \emptyset$ since $W_0 = \emptyset$, and
     \item[\eqref{it:attract}] 
       $E_0 = \{\Gamma^*\}$ is attracted to $I_0 = c[\nat]$.
  \end{enumerate}

  For the induction step, let $n\in\nat$ be arbitrary.
  We assume 
  $c[\overline{K}_{\!n}] \cap W_n = \emptyset$
  and that $E_n$ is attracted to~$I_n$ (induction hypothesis).
  We first prove the implication
  \begin{align}
    w_n \in c[\overline{K}_{\!n}] \implies w_n = c(k_n) \,.
    \label{eq:wn:in:Jn}
  \end{align}
  If $w_n \in c[\overline{K}_{\!n}]$,
  also $w_n \in ([w_n]_{\sim_n} \cap c[\overline{K}_{\!n}])$.
  Since $c[\overline{K}_{\!n}] \cap W_n = \emptyset$ by induction hypothesis,
  $w_n$ is the smallest element in $[w_n]_{\sim_n} \cap c[\overline{K}_{\!n}]$.
  Hence
  $w_n = \min([w_n]_{\sim_n} \cap c[\overline{K}_{\!n}])$ 
  and $c(k_n) = w_n$.

  \begin{enumerate}[(a)]
    \item 
      By the induction hypothesis we have $c[\overline{K}_{\!n}] \cap W_n = \emptyset$, 
      and so we have $c[\overline{K}_{\!n+1}] \cap W_n = \emptyset$.
      To see that $c[\overline{K}_{\!n+1}] \cap W_{n+1} = \emptyset$,
      it suffices to show $w_n \not\in c[\overline{K}_{\!n+1}]$.
      Assume, to derive a contradiction, that 
      $w_n \in c[\overline{K}_{\!n+1}]$.
      Then also 
      $w_n \in c[\overline{K}_{\!n}]$.
      Hence, by~\eqref{eq:wn:in:Jn}, we have $w_n = c(k_n)$.
      This contradicts $w_n \in c[\overline{K}_{\!n+1}]$,
      since, by injectivity of~$c$, 
      we have $c[\overline{K}_{\!n+1}] = c[\overline{K}_{\!n}]\setminus\{c(k_n)\} = c[\overline{K}_{\!n}]\setminus\{w_n\}$.

    \item 
      We have to prove that $E_{n+1}$ is attracted to $I_{n+1}$.
      We first show that 
      \begin{align}
        \text{$E_n$ is attracted to $I_{n+1}$}.
        \label{eq:En:att:In+1}
      \end{align}
      By injectivity of $c$ 
      we have $I_{n+1} = c[\overline{K}_{\!n}]\setminus\{c(k_n)\} \cup W_n \cup \{w_n\}$.
      Moreover, 
      since by induction hypothesis $c[\overline{K}_{\!n}] \cap W_n = \emptyset$,
      we have
      $c[\overline{K}_{\!n}]\setminus\{c(k_n)\} \cup W_n \cup \{w_n\} = (c[\overline{K}_{\!n}] \cup W_n)\setminus\{c(k_n)\} \cup \{w_n\}$. 
      Hence $I_{n+1} = (I_{n} \setminus \{c(k_{n})\}) \cup \{w_n\}$.
      Let $C \in E_n$ be arbitrary.
      By induction hypothesis we know that $C$ is attracted to $I_n$; 
      we distinguish the following two cases:
      \begin{itemize}
        \item 
          If $|C \cap I_{n}| = \infty$, then $|C \cap I_{n+1}| = \infty$.
          Hence $C$ is attracted to $I_{n+1}$.
        \item 
          Assume $C \subseteq I_{n}$.
          If $c(k_n) \not\in C$, then clearly $C \subseteq I_{n+1}$.
          If $c(k_n) \in C$, then $C = [w_n]_{\sim_n}$ since $c(k_n) \sim_n w_n$ due to \eqref{eq:ckn:sim:wn}. 
          From $w_n \in C \subseteq I_n$, we obtain $w_n \in c[\overline{K}_{\!n}]$.
          Hence, by \eqref{eq:wn:in:Jn} we have $w_n = c(k_n)$ and so $I_n = I_{n+1}$.
          
      \end{itemize}
      This concludes the proof of \eqref{eq:En:att:In+1}.

      Recall that $A'_n$ is the automaton $A'$
      obtained by invoking Lemma~\ref{lem:modify} with $A = A_n$, $I = I_{n+1}$, and $E = E_n$.
      Both requirements of the lemma are established above, 
      in \eqref{eq:reg} and \eqref{eq:En:att:In+1},
      so the automaton~$A'_n$ is well-defined.
      Let, moreover, $E'$ be the resulting partition obtained in the lemma.
      Lemma~\ref{lem:modify} guarantees that $E'$ is attracted to $I_{n+1}$.
      Moreover, by definition of $E_{n+1}$, we have $E' = E_{n+1}$.
      Hence $E_{n+1}$ is attracted to $I_{n+1}$.
      
  \end{enumerate}

  We now establish that
  every natural number $k$ will be picked as a $k_n$ eventually.
  That is, for every $k \in \nat$ there exists $n \in \nat$ such that $k_n = k$.
  Let $k \in \nat$.
  For every $n \in \nat$, the set $c[\overline{K}_{\!n}]$ 
  is the image under $c$ of all natural numbers that have not yet been picked
  as a $k_i$ for $0\le i\le n-1$.
  Since $w_0,w_1,\ldots$ is an enumeration of $\Gamma^*$,
  there exists $n \in \nat$ such that $c(k) \in W_n$.
  Hence $k \not\in \overline{K}_{\!n}$
  since $c[\overline{K}_{\!n}] \cap W_n = \emptyset$ by property~\eqref{it:intersect}.
  Thus $k$ will be picked eventually.
  
  We define the encoding $d : \nat \to \Gamma^*$ by
  \begin{align*}
    d(i) = \lim_{n \to \infty} d_n(i) && (i\in\nat)
  \end{align*}
  These limits are well-defined
  since, for every $i\in\nat$:
  there exists $n \in \nat$ with $i = k_n$,
  and we have
  $d_m(i) = c(i)$ for all $m \le n$ and
  $d_m(i) = w_n$ for all $m > n$.

  From the above, it follows that the function $d$ has the following property:
  \begin{align}
    d(k_n) = w_n \quad \text{for all $n\in\nat$}\;.
    \label{eq:d}
  \end{align} 
  We now show that $d$ is indeed a bijection.
  Since by construction, every $w_0,w_1,w_2,\ldots$ 
  occurs precisely once in the image of $d$,
  $d$ is injective.
  Moreover, we have seen above, that the set $\{k_0,k,k_2,\ldots\}$, the domain of $d$, contains all natural numbers.
  Hence $d$ is a bijection.
  
  Finally, we show that 
  for all $L \subseteq \nat$,
  if $c[L]$ is relatively regular in the image $c[\nat]$,
  then $d[L]$ is a regular language.
  Let $L \subseteq \nat$ be such that $c[L]$ is relatively regular in $c[\nat]$.
  Then there exists a regular language $R$
  such that ${c[L]} = {R \cap c[\nat]}$.
  Let $m \in \nat$ be such that 
  $R = {\lang{A_m}}$.
  Then for all $n \in \nat$,
  \begin{align}
    n \in L \iff c(n) \in \lang{A_m} \,. \label{eq:L:Am}
  \end{align}
  By the above construction we have for all $w\in I_{m}$
  \begin{align*}
    w \in \lang{A'_m} \iff w \in \lang{A_m} \,.
  \end{align*}
  By definition, 
  $I_{m}$ coincides with $c[\nat]$ for all but finitely many words.
  Hence for almost all $n \in \nat$
  \begin{align}
    c(n) \in \lang{A_m} \iff c(n) \in \lang{A'_m} \label{eq:Am:Am'}
    \,.
  \end{align}
  By the definition of $d$, and property~\eqref{eq:ckn:sim:wn}, we find: 
  \begin{align*}
    d(k_n) = w_n \mathrel{\sim_n} c(k_n) && \text{for all $n\in\nat$}\,.
  \end{align*}
  Due to ${\sim_n} \subseteq {\sim_m}$ for every $n \ge m$, 
  we obtain that 
  \begin{align*}
    d(k_n) \mathrel{\sim_m} c(k_n) &&\text{for all $n\ge m$}\,,
  \end{align*}
  and hence, since $\nat = \{ k_n \mathrel{:} n\in\nat \}$ holds,
  that 
  \begin{align*}
    d(i) \mathrel{\sim_m} c(i) && \text{for almost all $i\in\nat$}\,.
  \end{align*}
  Hence for almost all $n$ we have
  \begin{align*}
    d(n) \in \lang{A'_m} &\iff c(n) \in \lang{A'_m} \\
      &\iff c(n) \in \lang{A_m} && \text{by \eqref{eq:Am:Am'}} \\
      &\iff n \in L && \text{by \eqref{eq:L:Am}.} 
  \end{align*}
  As $d$ is bijective, we obtain
  \begin{align*}
    w \in \lang{A'_m} \iff w \in d(L) 
  \end{align*}
  for almost all $w \in \Gamma^*$.
  Hence $d(L)$ differs only by finitely many elements from a regular language
  and is consequently itself regular.

\end{proof}

The following proposition states that, under certain conditions,
the bijective function $g$ constructed in Lemma~\ref{lem:inj2bij}
is computable.
Obviously, the injective function $f$ that is lifted to $g$ must be computable to start with.
Moreover, we need to be able to decide for regular languages 
whether their intersection with the image of $f$ is empty, finite or infinite.
This enables us, in case of a finite intersection, to compute this intersection,
and to decide whether the equivalence classes $E_n$ are 
attracted to the image of $f$ (and $I_n$).
This suffices to ensure computability of $g$
constructed in the proof of Lemma~\ref{lem:inj2bij}.

\begin{proposition}\label{prop:computable}
  Let $A$ be a computable, countably infinite set and let $f : A \to \Gamma^*$ be a computable injection
  such that for every regular language~$R$ (given by an automaton), 
  emptiness and finiteness of the set  $R \cap f[\Gamma^*]$ is decidable.
  There exists a computable bijection $g : A \to \Gamma^*$
  such that for all $L \subseteq A$,
  if $f[L]$ is relatively regular in the image $f[A]$,
  then $g[L]$ is a regular language.
\end{proposition}

\begin{proof}
  By close inspection of the proofs of Lemmas~\ref{lem:modify} and~\ref{lem:inj2bij}
  it can be established that the construction of the encoding~$g$ from the encoding~$f$
  preserves computability under the stated conditions and by choosing computable enumerations
  $v_0,v_1,\ldots$ and $w_0,w_1,\ldots$ of $A$ and $\Gamma^*$, respectively.
\end{proof}

We are ready to state our main result.
\begin{theorem}\label{thm:main}
  Let $\Sigma$ and $\Gamma$ be finite alphabets, with $\card{\Gamma} \ge 2$.
  Let $\mcl{L} \subseteq \powerset{\Sigma^*}$ be a countable set of formal languages.
  There exists a bijective function $g : \Sigma^* \to \Gamma^*$ 
  such that for every $L \in \mcl{L}$,
  the image $g[L]$ is a regular language.
\end{theorem}
\begin{proof}
  By Lemmas~\ref{lem:inj:reg:crea} and~\ref{lem:inj2bij}.
\end{proof}

The following corollary justifies the title of this paper.

\begin{corollary}\label{cor:puzzle}
  There exists a computable bijective function $f : \Sigma^* \to \Gamma^*$
  such that the image function $f[\_]$ is regularity preserving,
  but the preimage function $f^{-1}[\_]$ is not. 
\end{corollary}

Without the requirement on the function $f$ to be computable, we could prove the statement as follows:
Let $\mcl{L} \subseteq \powerset{\Sigma^*}$ be the (countable) set of all recursive languages over $\Sigma$.
By Theorem~\ref{thm:main}
there exists a bijective mapping $f : \Sigma^* \to \Gamma^*$ such that 
$f[L]$ is regular for all $L \in \mcl{L}$.
Then, clearly, $f[\_]$ is regularity preserving while $f^{-1}[\_]$ is not.
However, the function $f$ obtained in this way is not computable. 
To obtain a computable function, we argue as follows.

\begin{proof}[Proof of Corollary~\ref{cor:puzzle}]
  Let $\Sigma = \{0,1\}$.
  We define a `balancedness' function $b : \Sigma^* \to \Sigma$ for all $w \in \Sigma^*$ by
  $b(w) = 1$ if the word $w$ contains an equal number of zeros and ones,
  and $b(w) = 0$ otherwise.
  Then we define 
  $f : \Sigma^* \to \Sigma^*$ by
  \begin{align*}
    f(w) = b(w) w 
  \end{align*}
  for every $w \in \Sigma^*$.
  Then we have:
  \begin{enumerate}
    \item 
      For every regular language $L \subseteq \Sigma^*$,
      we have that $f[L]$ is relatively regular in $f[\Sigma^*]$ 
      (an automaton can simply ignore the first letter).
    \item 
      However, the function $f^{-1}[\tvs]$ is not preserving (relative) regularity.
      To see this, let $X = \{w \mid b(w) = 1\}$.
      Clearly $f^{-1}[f[X]] = X$ is not regular.
      But $f[X]$ is relatively regular in $f[\Sigma^*]$, 
      since $f[X]$ consists precisely of those words in $f[\Sigma^*]$ that start with letter $1$.
  \end{enumerate}

  We now invoke Proposition~\ref{prop:computable} for lifting $f$ to a computable, bijective $g$.
  The conditions of the proposition are satisfied since $f$ is computable,
  and the image $f[\Sigma^*]$ of $f$ is a context-free language.
  The intersection of context-free languages with regular languages
  is context-free, and finiteness and emptiness are decidable.
  The proposition guarantees the existence of a computable, bijective function $g : \Sigma^* \to \Sigma^*$
  such that for every $L \subseteq \Sigma^*$ that is relatively regular in $f[\Sigma^*]$ 
  we have that $g[L]$ is regular.
  Then by (i) we have that $g[\tvs]$ preserves regularity.
  From (ii) it follows that $g[X]$ is regular while $g^{-1}[g(X)] = X$ is not.
  Hence $g^{-1}[\tvs]$ does not preserve regularity.
\end{proof}

We strengthen the statement of Corollary~\ref{cor:puzzle} by extending it to the preservation
of membership in countable classes of languages that include the regular languages.  
For this, we use the following stipulation.
For an arbitrary set~$\mathcal{S}$ of languages over $A$, we say
that a function $F : \powerset{A^*} \to \powerset{B^*}$ preserves membership in~$\mathcal{S}$
if, for all languages over $A$, $L\in\mathcal{S}$ implies $F(L)\in\mathcal{S}$. 

\begin{corollary}
  Let $\mathcal{S} \subseteq \powerset{\Sigma^*}$ be a countable set of languages that includes the regular languages.
  Then there exists a bijective function $f : \{0,1\}^* \to \{0,1\}^*$
  such that $f[\_]$ preserves membership in $\mathcal{S}$, 
  but $f^{-1}[\_]$ does not.
\end{corollary}

\begin{proof}
  Let $\mathcal{S}' = \mcl{S} \cup \{X\}$ where $X$ is a language
  not in $\mcl{S}$.
  By Theorem~\ref{thm:main}
  there exists a bijective function $f : \Sigma^* \to \Gamma^*$ such that 
  $f[L]$ is regular for all $L \in \mcl{S}'$.
  Then, $f[\_]$ preserves membership in $\mcl{S}$ while $f^{-1}[\_]$;
  note that
  $f(X)$ is regular and hence $f(X) \in \mcl{S}$, but $f^{-1}[f(X)] = X \not\in \mcl{S}$.
\end{proof}

\section{Consequences for Comparing \\ Models of Computation}\label{sec:moc}

It turns out that our main results have some remarkable consequences in the context
of comparing computational models 
using concepts proposed by Boker and Dershowitz
in a series of publications \cite{boke:2004,boke:ders:2006,boke:2008}.
The authors summarize their goal as follows:  
\begin{quotation}
  ``We seek a robust definition of relative power that does not itself depend 
    on the notion of computability. 
    It should allow one to compare arbitrary models over arbitrary domains 
    via a quasi-ordering that successfully captures the intuitive concept 
    of computational strength.
    [\ldots]\footnote{%
      This passage continues:
        ``Eventually, we want to be able to prove statements like 
        `analogue machines are strictly more powerful than digital devices', 
        even though the two models operate over domains of different cardinalities.''}\hspace*{1pt}''
     \cite{boke:ders:2006}     
\end{quotation}

This motivation leads them to specific choices of simple and liberal 
conditions on encodings.
Encodings are typically used to translate between models of computation that act on different domains.
So an encoding~$\rho : D_1 \to D_2$ can facilitate 
the simulation of the input-output behavior of a machine $M_1$ belonging to a model $\mcl{M}_1$ with domain~$D_1$
by a machine $M_2$ from a model $\mcl{M}_2$ with domain~$D_2$\,:
\begin{center}
  \begin{tikzpicture}[baseline=2ex,node distance=25mm]
    \node (l) {$D_1$};
    \node (r) [right of=l] {$D_2$};
    \node (l') [below of=l,node distance=15mm] {$D_1$};
    \node (r') [right of=l'] {$D_2$};
    \draw[->] (l) -- node [above] {$\rho$} (r);
    \draw[->] (l') -- node [below] {$\rho$} (r');
    \draw[->] (l) -- node [left] {$M_1$} (l');
    \draw[->] (r) -- node [right] {$M_2$} (r');
  \end{tikzpicture}
\end{center}

To prevent codings from participating too strongly in
the simulation of a computation on a machine $M_1$ 
through a computation on a machine $M_2$ 
(and thereby from substantially alleviating, for the simulating machine $M_2$, 
the task that is solved by the simulated machine $M_1$),
codings are usually required to be computable in some sense.
Frequently, one of the following two restrictions are stipulated
(see for example Rogers' classic book \cite[p.\hspace*{1pt}27,\hspace*{1pt}28]{roge:1967}):
\begin{enumerate}\renewcommand{\labelenumi}{(\arabic{enumi})}
  \item{}\label{coding:effective}
    Codings must be `informally algorithmic', `informally computable',
    or `effective' in the sense that they can be carried out by an in principle mechanizable procedure.     
  \item{}\label{coding:computable}
    Codings are required to be computable with respect to a specific model,
    for example by a Turing machine.  
\end{enumerate}
Boker and Dershowitz reject such prevalent stipulations:
\begin{quotation}
     ``Effectivity is a useful notion; however, it is unsuitable as a general power comparison notion.
     The first, informal approach is too vague, while the second can add computational power 
     when dealing with subrecursive models and is inappropriate when dealing with non-recursive models.''
     \cite{boke:2008}
\end{quotation}
As a consequence, they go on to use classes of encodings 
that do not constrain (at least not explicitly) the cost that is necessary to compute an encoding.  
In particular, they define three concepts of comparison 
(see Definition~\ref{def:mod:sim} below)
that are, respectively, based on:\hspace*{0.5pt}%
  \footnote{%
    Note that encodings of the notions~(i) and (ii) here are more liberal
    than those in (1) 
                  and (2) 
    above, and that therefore the use of such encodings does not address 
    the concern raised in the preceding quotation 
    regarding the addition of computational power when dealing with subrecursive models.} 
\begin{enumerate}
  \item 
    encodings (injective functions) without any additional requirement;
  \item     
     encodings that are `decent' with respect to the simulating model~$\mcl{M}_2$ that shoulders the simulation,
     in the sense that $\mcl{M}_2$ is able to recognize the image of the coding;
   \item 
     bijective encodings.
\end{enumerate}

We will show that each of these concepts admits some quite counterintuitive consequences.
At first these anomalies pertain only to decision models,
the subclass of all models that only obtain `yes'/`no' as computation result.
But it turns out these phenomena apply also to more broad classes of models.

In order to formally state our results, we repeat here the basic definitions in \cite{boke:ders:2006,boke:2008},
and extend them by straightforward adaptations for decision models. 

By abstracting away from all intensional aspects of models of computation
that concern mechanistic aspects of stepwise computation processes, 
Boker and Dershowitz define a model extensionally as an arbitrary set of 
(extensionally represented) partial functions over some domain.  
  
\begin{definition}[{\cite[Def.~2.1]{boke:ders:2006}}]\label{def:moc}
  A \emph{model of computation} is a pair $\mcl{M} = \pair{D}{\mcl{F}}$,
  where $D$ is a set of elements, 
  the \emph{domain of $\mcl{M}$},
  and $\mcl{F} = \{ f \where f : D \to D \cup \{ \bot \} \}$ is a set of functions
  with $\bot \not\in D$.
  We write $\mrm{dom}_\mcl{M}$ for the domain of $\mcl{M}$.
  (We assume that $\bot$ is a fixed element not contained in the domain of any model.)
\end{definition}

We define decision models as models of computation 
consisting of total functions that yield a definite `yes'/`no' answer.

\begin{definition}\label{def:mod}
  A \emph{decision model} (model of computation for decision models) 
  is a model of computation $\mcl{M} = \pair{D}{\mcl{F}}$,
  such that $\{0,1\} \subseteq D$,
  and $f[D] \subseteq \{0,1\}$, for all $f \in \mcl{F}$.
\end{definition}

Now codings between domains of models are defined. 

\begin{definition}[{\cite[Def.~2.2]{boke:ders:2006}}]\label{def:coding}
  Let $D_1$ and $D_2$ be domains of models of computation.
  A \emph{coding (from $D_1$ to $D_2$)} 
  is an injective function $\rho : D_1 \cup \{\bot\} \to D_2 \cup \{\bot\}$
  such that $\rho(x) = \bot$ if and only if $x = \bot$, for all $x \in D_1$.
\end{definition}

For codings $\rho$ between decision models it could be desirable to
demand that $\rho(0) = 0$ and $\rho(1) = 1$. 
We do not to take up this restriction, for a pragmatic reason
connected to the definition of `simulation' immediately below. 
If namely a non-constant function $f$ in a decision model $\mcl{M}_1$
is simulated via $\rho$
by a function $g$ in a decision model $\mcl{M}_2$,
then it follows that either $\rho(0) = 0$ and $\rho(1) = 1$, 
or $\rho(0) = 1$ and $\rho(1) = 0$.
In both cases it can be said that decisions taken by $f$
are faithfully modelled by corresponding decisions taken by~$g$.  

\begin{definition}[{\cite[Def.~2.3]{boke:ders:2006}}]\label{def:fun:sim}
  Let $\mcl{M}_1 = \pair{D_1}{\mcl{F}_1}$
  and $\mcl{M}_2 = \pair{D_2}{\mcl{F}_2}$
  be models of computation.
  Let $\rho$ be a coding from $D_1$ to $D_2$.
  We define:
  \begin{enumerate}
    \item 
      For $g \in \mcl{F}_2$ and $f \in \mcl{F}_1$ we say that 
      $g$ \emph{simulates} $f$ \emph{via} $\rho$
      if $g \circ \rho = \rho \circ f$ holds, 
      as in the following diagram:
      \begin{center}
        \begin{tikzpicture}[baseline=2ex,node distance=25mm]
          \node (l) {$D_1$};
          \node (r) [right of=l] {$D_2$};
          \node (l') [below of=l,node distance=15mm] {$D_1$};
          \node (r') [right of=l'] {$D_2$};
          \draw[->] (l) -- node [above] {$\rho$} (r);
          \draw[->] (l') -- node [below] {$\rho$} (r');
          \draw[->] (l) -- node [left] {$f$} (l');
          \draw[->] (r) -- node [right] {$g$} (r');
        \end{tikzpicture}
      \end{center}

    \item 
      $\mcl{M}_2$ \emph{simulates} $\mcl{M}_1$ \emph{via} $\rho$,
      denoted by $\mcl{M}_1 \lesssim_\rho \mcl{M}_2$,
      if for every $f \in \mcl{F}_1$ there is a
      $g \in \mcl{F}_2$ such that $g$ simulates $f$ via $\rho$.
  \end{enumerate}
\end{definition}

The `decency' requirement for codings mentioned before is defined
as follows in \cite{boke:2008}. There, Boker stresses that this requirement
follows classic definitions of computable groups 
by Rice \cite[p.\hspace*{1.5pt}298]{rice:1956} and Rabin \cite[p.\hspace*{1.5pt}343]{rabi:1960}.

\begin{definition}[{\cite[Def.~52]{boke:2008}}]\label{def:decent}
  Let $\mcl{M}_1 = \pair{D_1}{\mcl{F}_1}$
  and $\mcl{M}_2 = \pair{D_2}{\mcl{F}_2}$
  be models of computation.
  A coding $\rho$ from $D_1$ to $D_2$
  is called \emph{decent} with respect to $\mcl{M}_2$
  if the image $\rho[D_1]$ can be recognized by $\mcl{M}_2$, 
  in the sense that
  there is a total function~$g \in \mcl{F}_2$ 
  such that 
  $\rho[D_1] = g[D_2]$ and
  for all $y \in D_2$, we have $g(y) = y$ if and only if $y \in \rho[D_1]$.
\end{definition}

We note that according to Definition~\ref{def:decent}
a coding~$\rho$ 
can be decent with respect to a decision model~$\mcl{M}_2$ only if
$\rho[D_1] \subseteq \{0,1\}$.
Since $\rho$ is injective, $\card{D_1} \le 2$ follows, and so $\mcl{M}_1$ can only be a rather trivial model.   
Therefore we adapt the notion of decency in an obvious way
to accommodate decision models.

\begin{definition}\label{def:decent:*}
  Let $\mcl{M}_1 = \pair{D_1}{\mcl{F}_1}$
  and $\mcl{M}_2 = \pair{D_2}{\mcl{F}_2}$
  be models of computation.
  A coding $\rho$ from $D_1$ to $D_2$
  is called \emph{decent$^*$} with respect to $\mcl{M}_2$
  if the image $\rho[D_1]$ can be recognized by $\mcl{M}_2$, 
  in the sense that
  there is a total function~$g \in \mcl{F}_2$ and an element $d \in D_2$ 
  such that for all $y \in D_2$, we have $g(y) = d$ if and only if $y \in \rho[D_1]$.
\end{definition}

With the concepts `model of computation' and `simulation' defined,
Boker and Dershowitz introduce three comparison preorders for models,
which are based on three classes of codings as mentioned above.
In addition to the preorder induced by decent codings,
we also define a variant preorder induced by decent$^*$ codings.

\begin{definition}[{\cite[Def.\hspace*{1.5pt}52]{boke:2008}}]\label{def:mod:sim}
  Let $\mcl{M}_1 = \pair{D_1}{\mcl{F}_1}$
  and $\mcl{M}_2 = \pair{D_2}{\mcl{F}_2}$
  be models of computation.
  We define:
  \begin{enumerate}
    \item 
      $\mcl{M}_2$ \emph{is at least as powerful as} $\mcl{M}_1$,  
      denoted by 
      \begin{center}
        $\mcl{M}_1 \lesssim \mcl{M}_2$\,,
      \end{center}
       if $\mcl{M}_1 \lesssim_\rho \mcl{M}_2$ for some $\rho$\,.
       \\[-1.5ex]
      
    \item 
      $\mcl{M}_2$ \emph{is at least as powerful as} $\mcl{M}_1$ \emph{via a decent coding},
      which we denote by 
      \begin{center}
        $\mcl{M}_1 \lesssim_{\mrm{decent}} \mcl{M}_2$\,, 
      \end{center}
      if $\mcl{M}_1 \lesssim_\rho \mcl{M}_2$ for some decent coding~$\rho$ with respect to $\mcl{M}_2$.
      %
      $\mcl{M}_2$ \emph{is at least as powerful as} $\mcl{M}_1$ \emph{via a decent$^*$ coding},
      which we denote by 
      \begin{center}
        $\mcl{M}_1 \lesssim_{\mrm{decent}^*} \mcl{M}_2$\,,  
      \end{center}
      if $\mcl{M}_1 \lesssim_\rho \mcl{M}_2$ for some decent$^*$ coding~$\rho$ with respect to $\mcl{M}_2$.
       \\[-1.5ex]

    \item
      $\mcl{M}_2$ \emph{is at least as powerful as} $\mcl{M}_1$ \emph{via a bijective coding},
      which we denote by 
      \begin{center}
        $\mcl{M}_1 \lesssim_{\mrm{bijective}} \mcl{M}_2$\,,
      \end{center}
      if $\mcl{M}_1 \lesssim_\rho \mcl{M}_2$ for some bijective coding $\rho$\,.
  \end{enumerate}
\end{definition}

With these definitions in place, we are now able to state, and prove,
our results concerning the comparison of decision models.
For this we denote, for an alphabet $\Gamma$ with $\setexp{0,1} \subseteq \Gamma$,
by $\mocDFAsover{\Gamma} = \pair{\Gamma^*}{\mcl{D}}$ with
\begin{center}
  $
  \mcl{D} = \descsetexpBig{f \funin \Gamma^*\to\setexp{0,1}}
                          {\parbox{32ex}{$\exists M \text{ DFA} \: \forall w\in\Gamma^*$\\[-0.1  ex]
                                          \hspace*{3ex}$(\, \fap{f}{w} = 1 \Leftrightarrow \text{$M$ accepts $w$}\,)$}}
  $
\end{center}
the decision model consisting 
of all functions $f \funin \Gamma^* \to \setexp{0,1}$ that 
describe the ac\-cep\-tance\discretionary{/}{}{/}non-ac\-cep\-tance behavior 
of a DFA with input alphabet $\Gamma$.
The model $\mocTMDsover{\Gamma}$ over input alphabet $\Gamma$ of Turing-machine deciders is defined analogously.

The proposition below is an easy consequence of Lemma~\ref{lem:inj:reg:crea}.

\begin{proposition}\label{prop:moc:inj}
  Let $\Sigma$, $\Gamma$ be alphabets, where $\setexp{0,1} \subseteq \Gamma$. 
  Then for every countable decision model $\mcl{M}$ with domain~$\Sigma^*$
  \begin{equation}\label{eq:prop:moc:inj}
    \mcl{M} \; \lesssim \mocDFAsover{\Gamma} \,,
  \end{equation}
  holds, that is, deterministic finite state automata with input alphabet~$\Gamma$
  are at least as powerful as $\mcl{M}$.
\end{proposition}

\begin{proof}
  Every decision model $\mcl{M} = \pair{\Sigma^*}{\mcl{F}}$ with domain $\Sigma^*$
  and with countable set $\mcl{F}$ of computed functions
  corresponds to the countable set 
  $\mcl{L}_{\mcl{M}}
     =
   \descsetexp{ L_f }{ f\in\mcl{F} }  $
  of languages   
  that are defined, for $f\in\mcl{F}$, as
  $ L_f = \descsetexp{w\in\Sigma^*}{ \fap{f}{w} = 1 }$.   
  By Lemma~\ref{lem:inj:reg:crea} 
  there exists an injective function $\rho \funin \Sigma^* \to \Gamma^*$
  such that $\rho[L_f]$ is relatively regular in $\rho[\Sigma^*]$, for all $f\in\mcl{F}$.
  Now it is straightforward to verify that $\rho$ is a coding between the domains 
  of the models $\mcl{M}$ and $\mocDFAsover{\Gamma}$
  that facilitates the simulation of every $f\in\mcl{F}$ 
  by a function $g \funin \Gamma^* \to \setexp{0,1}$
  that denotes the ac\-cep\-tance\discretionary{/}{}{/}non-ac\-cep\-tance behavior 
  of a deterministic finite-state automaton with input alphabet $\Gamma$.
  This shows \eqref{eq:prop:moc:inj}.
\end{proof}

\begin{proposition}\label{prop:moc:inj:TM}
  Let $\Gamma$ be an alphabet with $\setexp{0,1} \subseteq \Gamma$. 
  Then there is a coding $\rho \funin \Gamma^*\to\Gamma^*$ 
  such that 
  $\mocTMDsover{\Gamma} \lesssim_{\rho} \mocDFAsover{\Gamma}$ holds.
  But any such a coding $\rho$ cannot be computable.
\end{proposition}

\begin{proof}
  The main statement follows from Proposition~\ref{prop:moc:inj}.
  That $\rho$ cannot be computable can be seen as follows.
  Suppose that $\rho$ is computable.
  Let $A_0,A_1,\ldots$ and $w_0,w_1,\ldots$ be recursive enumerations 
  of all finite automata and words over~$\Sigma$.
  Then the language 
                    $L = \descsetexp{w_n}{n \in \nat,\, \rho(w_n)\notin\lang{A_n}}$
  is Turing computable, 
  but $\rho[L]$ is not regular.
\end{proof}

This statement can be strengthened to a bijective, and therefore (see the proof) also decent$^*$, 
simulation with finite-state automata
by using Lemma~\ref{lem:inj2bij} and Theorem~\ref{thm:main}.
\begin{corollary}\label{cor:moc:bij}
  Let $\Sigma$ and $\Gamma$ be alphabets, where $\setexp{0,1} \subseteq \Gamma$. 
  Then for every countable decision model $\mcl{M}$ 
  with domain $\Sigma^*$ the following two statements hold:
  \begin{enumerate}\setlength{\itemsep}{0.5ex}
      \vspace*{0.25ex}
    \item{}\label{item:i:cor:moc:bij}
       $\mcl{M} \; \lesssim_{\mrm{bijective}} \mocDFAsover{\Gamma}$,
    \item{}\label{item:ii:cor:moc:bij} 
      $\mcl{M} \; \lesssim_{\mrm{decent}^*} \mocDFAsover{\Gamma}$.
  \end{enumerate}
  That is, deterministic finite-state automata with input alphabet~$\Gamma$
  are at least as powerful as the model $\mcl{M}$, both via bijective and via decent$^*$ codings.
\end{corollary}

\begin{proof}
  Statement~\ref{item:i:cor:moc:bij} follows by 
  an argumentation analogous to the proof of Proposition~\ref{prop:moc:inj}
  in which the use of Lemma~\ref{lem:inj:reg:crea} 
  is replaced by an appeal to our main theorem, Theorem~\ref{thm:main}. 
  
  Statement~\ref{item:ii:cor:moc:bij}
  follows directly from statement~\ref{item:i:cor:moc:bij}:
  For bijective codings $\rho : \Sigma^* \to \Gamma^*$,
  the image of $\rho$ coincides with $\Gamma^*$, which is trivially 
  recognizable by a finite-state automaton. 
\end{proof}   
      
\begin{remark}      
  As mentioned above (just before Definition~\ref{def:decent:*}), 
  decent codings in the sense of~\cite{boke:2008}
  do not form a sensible notion for decision models.
  However, for every countable decision model $\mcl{M}$
  and every alphabet $\Gamma$ with $\setexp{0,1}\subseteq\Gamma$  
  it holds that
  \begin{center}
    $
    \mcl{M} \; \lesssim_{\mrm{decent}} \mocDFAswithidover{\Gamma} 
    $
  \end{center}
  where $\mocDFAswithidover{\Gamma}$
  is the extension of $\mocDFAswithidover{\Gamma}$ by adding the identity function $\mrm{id} \funin \Gamma^* \to \Gamma^*$.
 
  Sequential finite-state transducers $\mocFSTsover{\Gamma}$ (see e.g.~\cite{saka:2009}) over alphabet $\Gamma$
  form a natural computational model that extends  $\mocDFAswithidover{\Gamma}$. 
  Thus for an alphabet $\Gamma$ with $\setexp{0,1}\subseteq\Gamma$ we also obtain the very counterintuitive result: 
  \begin{center}
    $\mcl{M} \; \lesssim_{\mrm{decent}} \mocFSTsover{\Gamma} \,,$
  \end{center}
  that is,  
  sequential finite-state transducers are at least as strong via a decent coding as
  every countable decision model.
  (Here we consider finite-state transducers that are able to recognize the end of a word.)
  In particular, every Turing-machine decider 
  can be simulated by a sequential finite state transducer via a decent coding.
\end{remark} 

\begin{remark}\label{rem:gen:models}
  These results raise the question, whether these anomalies only concern decision models. 
  In particular, 
  one may wonder
  whether the comparison of \emph{computational models} 
  avoids counterintuitive results 
  when additional requirements are imposed on the models that are compared.
  A candidate requirement would be to enforce that the output of the models must have an infinite range.
  Or, even stronger, we could require the following property: 
%
    A class of models $\mcl{M} = \pair{D}{\mcl{F}}$ is 
    \emph{image-complete} if for every non-empty computable 
    set $I \subseteq D$ there exists $f \in \mcl{F}$ such that $f[D] = I$.
  
  Let $\mcl{T} = \pair{\Sigma^*}{\mcl{F}}$ consist of all Turing machines
  such that for every $f \in \mcl{F}$ we have
  there exists a finite set $L_f \subseteq \Sigma^*$ and for all $x \in \Sigma^*$ we have
  $f(x) \in \{x\} \cup L_f \cup \{\bot\}$\,.
  Thus functions $f \in \mcl{T}$ map words either to themselves or into a finite set
  that may depend on $f$.
  The class $\mcl{T}$ is a natural model because
  it can be implemented by a recursively enumerable set of Turing machines\footnote{%
    The idea is to enumerate all Turing machines and finite sets $L_f$, and adapt
    the machines to check on termination whether the output is in $L_f$
    and otherwise make sure that the output is equal to the input.
    In this way, we obtain all machines that are necessary to implement $\mcl{T}$.}.
  Note that $\mcl{T}$ is a strict extension of Turing-machine deciders
  and it is an image-complete model.
  This can be seen as follows:
  Let $I \subseteq \Sigma^*$ be any non-empty computable set.
  Let $i \in I$ and define the function $f : \Sigma^* \to \Sigma^*$  for all $x \in \Sigma^*$ by:
  $f(x) = x$ if $x \in I$ and $f(x) = i$, otherwise.
  Then $f[\Sigma^*] = I$ and $f \in \mcl{T}$.
  
  Then the model $\mcl{T}$ 
  can be simulated by two-way sequential, finite-state transducers~\cite{enge:hoog:2001}
  via decent codings:
  \begin{center}
    $\mcl{T} \; \lesssim_{\mrm{decent}} \moctwowayFSTsover{\Gamma}\,.$
  \end{center}
  We have already argued that the Turing-machine deciders can be simulated
  by $\mocFSTsover{\Gamma}$.
  This can easily be generalized to a finite number of output words $L$
  (finite-state transducers can output words instead of only symbols).
  Now we assume that one symbol $w \in L$ symbolizes the identity output;
  then the two-way finite-state transducer, instead of producing this output word,
  can walk back to the beginning of the input and reproduce the input word as output word.
\end{remark}

Our results suggest that there are definite limitations 
to the concepts of power comparison for models of computation by Boker and Dershowitz.
These concepts have an `absolute' flavor insofar as 
they do not formulate any explicit constraints on the computability of encodings used for simulations.
The counterintuitive consequences 
                                  pertain primarily to decision models
(yet this is a blurry concept, see Remark~\ref{rem:gen:models}), 
and do not extend to models that include all partial-recursive functions (see Corollary~\ref{cor:thm:shapiro:extended} below). 
Yet they demonstrate that these comparison concepts
lack the desired robustness. 

We note that our results are not the first indications of anomalies.
In~\cite[Example~5.1]{boke:ders:2006}
Boker and Dershowitz show that Turing-machine deciders are not a complete model of computation,
in the sense that this model can be strictly extended to incorporate a non-recursive set.
Our results strengthen this example naturally in the following three ways:
(i)~to bijective, and decent encodings, (ii)~to use finite automata instead of Turing-machine deciders,
and (iii)~to arbitrary countable models as extensions.
This is because, as we have shown, finite automata can be extended, via decent codings, 
to any countable decision model, and consequently the same holds for Turing-machine deciders.
Hence, even decent encodings facilitate the simulation of all Turing-machine deciders 
by finite automata, more precisely, by finite-state transducers.

The following theorem is an easy consequence of the concepts developed by Shapiro in~\cite{shap:1982}.
He calls a number representation $r : \nat \to \Sigma^*$ \emph{acceptable}
if it is bijective, and if the successor function $\msf{succ} \funin \nat \to \nat$ 
can be simulated by a Turing machine on the representations,
that is, if the lifting $\msf{succ}^r : \Sigma^* \to \Sigma^*$ of   $\msf{succ}$
with the property:
\begin{align*}
  \msf{succ}^r(r(n)) = r(\msf{succ}(n))\,, && \text{for all $n\in\nat$,}
\end{align*}
is Turing computable.

\begin{theorem}\label{thm:shapiro:extended}
  A bijective encoding $f : \Sigma^* \to \Gamma^*$
  is computable
  if and only if
  there is an acceptable number representation $r$
  such that $f \circ r$ is acceptable as well.
\end{theorem}

This theorem implies that Corollary~\ref{cor:moc:bij}~\ref{item:i:cor:moc:bij}
does not generalize to models that compute the partial-recursive functions.

\begin{corollary}\label{cor:thm:shapiro:extended}
  There is no bijective encoding $\rho \funin \Sigma^*\to\Sigma^*$
  such that $\mocTMsover{\Sigma} \lesssim_{\rho} \moctwowayFSTsover{\Sigma}$ holds,
  where $\mocTMsover{\Sigma}$ is the model of Turing machines over alphabet $\Sigma$.
\end{corollary}

While certainly more investigation is needed,
we also interpret our results as follows.
For comparing the computational power of models of computation over different domains,
it is crucial to make clear how computational power should be measured for the purpose at hand.
After having settled on a reasonable measure,
this measure can then be used to constrain the computational power of admissible codings
that may act as a trustworthy intermediary between the models.   

\section{Consequences for Generalized Automaticity}\label{sec:auto}
Finite-state automata can be used to generate infinite sequences, see~\cite{allo:shal:2003}.
This is usually done using the standard base-$k$ representation $(n)_k \in \{0,1,\ldots,k-1\}^*$ of the natural numbers $n \in \nat$.
A sequence $\sigma\in\str{\Delta}$ 
is called \emph{$k$-automatic} if, 
for some deterministic finite-state automaton $A$ with output (DFAO, see Section~\ref{sec:prelims})
we have $\sigma(n) = A((n)_k)$ for all $n \in \nat$.

This concept has been generalized in several ways,
where different number representations are fed to the automaton,
see, e.g., \cite{shal:1989,rigo:2000,endr:grab:hend:2013}.
This motivates the study of automaticity with respect to arbitrary number representations,
which is part of work in progress of the present authors with Clemens Kupke, 
Larry Moss, and Jan Rutten.

\begin{definition}\label{def:auto}
  Let $c : \nat \to \Gamma^*$ be a number representation.
  A sequence~$\sigma \in \str{\Delta}$ over an alphabet $\Delta$
  is \emph{$c$-automatic} if there exists a DFAO~$A$
  such that $\sigma(n) = A(c(n))$ for all $n \in \nat$.
\end{definition}

\begin{lemma}\label{lem:fiber}
  Let $c : \nat \to \Gamma^*$ be a number representation.
  A sequence~$\sigma \in \str{\Delta}$ 
  is $c$-automatic if and only if
  for every $a \in \Delta$ the `fiber' 
  $\{ c(n) \where \sigma(n) = a \}$
  is relatively regular in $c[\nat]$.
\end{lemma}
\begin{proof}
  Along the lines of Lemma~5.2.6 in~\cite{allo:shal:2003}.
\end{proof}

\begin{corollary}\label{cor:inj2bij:auto}
  For every injective function $c : \nat \to \Gamma^*$ 
  there is a bijective function $d : \nat \to \Gamma^*$
  such that for every  $\sigma \in \str{\Delta}$ we have:
  if $\sigma$ is $c$-automatic, then $\sigma$ is also $d$-automatic.
\end{corollary}
\begin{proof}
  Let $c : \nat \to \Gamma^*$ be an injection.
  Let $d : \nat \to \Gamma^*$ be the bijective function obtained from $c$ by Lemma~\ref{lem:inj2bij}.
  Let $\sigma \in \str{\Delta}$ be $c$-automatic.
  We show that $\sigma$ is $d$-automatic by an application of Lemma~\ref{lem:fiber}.
  Let $a \in \Delta$. 
  Then the fiber
  $\{ c(n) \where \sigma(n) = a \}$
  is relatively regular in $c[\nat]$ by Lemma~\ref{lem:fiber}.
  By Lemma~\ref{lem:inj2bij}
  we obtain that the fiber $\{ d(n) \where \sigma(n) = a \}$ is regular
  (and consequently relatively regular in $d[\nat] = \Gamma^*$).
  Hence, by Lemma~\ref{lem:fiber}, $\sigma$ is $d$-automatic.
\end{proof}

The following proposition shows that the implication in Corollary~\ref{cor:inj2bij:auto} 
cannot be strengthened to equivalence of $c$- and $d$-automaticity:
not for all injective $c$ there is a bijective $d$ 
so that $c$-automaticity and $d$-automaticity coincide.

\begin{proposition}\label{prop:counterexample:iff:thm:inj:bij}
  Define the representation $c : \nat \to \{0,1\}^*$ by $c(n) = 0^{n!}$.
  Then we have:
  \begin{enumerate}
    \item 
        For every sequence $\sigma \in \Delta^\nat$,
        $\sigma$ is $c$-automatic if and only if $\sigma$ is ultimately constant.
    \item 
        For every bijection $d : \nat \to \{0,1\}^*$,
        there is a $d$\nb-auto\-ma\-tic sequence that is not ultimately constant.
  \end{enumerate}
\end{proposition}

\begin{proof}
  We first prove the two implications of~(i).
  \begin{itemize}
    \item[($\Rightarrow$)]
      Let $\sigma \in \Delta^\nat$ be $c$-automatic.
      That is, for some automaton $A$, 
      $\sigma(n) = A(c(n))$ for all $n\in\nat$. 
      As there are finitely many states in $A$, 
      there exists $n_0,\ell \in \nat$ such that 
      $\ell > 0$ and $\delta(q_0,0^{n}) = \delta(q_0,0^{n+\ell})$ for all $n \ge n_0$,
      where $\delta$ is the transition function of $A$ and $q_0$ its starting state.
      Let $m_0 \in \nat$ be the smallest integer such that $m_0! \ge n_0 + \ell$.
      Then we have
      $\delta(q_0,0^{m!}) = \delta(q_0,0^{m_0!})$ for all $m \ge m_0$.
      The reason is that $m_0! \ge n_0$ and, for all $m \ge m_0$, 
      $m!$ is a multiple of $\ell$ (so that $m! = m_0! + k\ell$ for some $k\in\nat$).
    \item[($\Leftarrow$)]
      Let $\sigma \in \Delta^\nat$ be ultimately constant,
      that is, there exists $n_0 \in \nat$ such that 
      $\sigma(n) = \sigma(n_0)$ for all $n \ge n_0$.
      Let $m = n_0!$\,. 
      We define an automaton $A$ with states $q_0,q_1,\ldots,q_m$,
      and transition function $\delta$ defined by
      $\delta(q_i,0) = q_{i+1}$ for all $i \in \{0,1,\ldots,m-1\}$
      and $\delta(q_m,0) = q_m$.
      For the output of $q_i$ ($0 \le i \le m$) 
      we take $\sigma(j)$ if $i = j!$\,. 
      The output of the other states is irrelevant.
      Clearly we now have $\sigma(n) = A(c(n))$ for all $n \in \nat$, 
      and so $\sigma$ is $c$\nb-automatic. 
  \end{itemize}

  For (ii), let $d : \nat \to \{0\}^*$ be a bijective encoding.
  Let $A$ be an automaton such that $A(0^{2n}) = 0$ and 
  $A(0^{2n+1}) = 1$ for all $n \in \nat$,
  and define $\sigma \in \{0,1\}^\nat$ by $\sigma(n) = A(d(n))$ for $n \in \nat$.
  Note that $\sigma$ is $d$-automatic by definition.
  As $d$ is bijective, 
  there are infinitely many $m \in \nat$ such that $d(m)$ is of the form $0^{2n}$,
  and there are infinitely many $m \in \nat$ such that $d(m)$ is of the form $0^{2n+1}$.
  Hence $\sigma$ is not ultimately constant.
\end{proof}

The following corollaries are reformulations of our main result, Theorem~\ref{thm:main},
for recognizability and, more generally, automaticity, respectively.

\begin{corollary}\label{cor:main:recog}
    For every countable class~$\mcl{C}$ of sets of natural numbers 
    there is a bijective function $r : \nat \to \Sigma^*$
    such that every $S \in \mcl{C}$ 
    is \rrecognizable{} (i.e., there is a finite automaton deciding membership $n \in S$ on the input of $r(n)$).
\end{corollary}

\begin{corollary}\label{cor:main:auto}
  For every countable class $\mcl{S} \subseteq \str{\Delta}$ of infinite sequences over a finite alphabet~$\Delta$,
  there exists a bijective number representation~$d : \nat \to \Gamma^*$
  such that every $\sigma \in \mcl{S}$ is $d$-automatic.
\end{corollary}

\section{Conclusion and Further Questions}\label{sec:questions}
Our main result, Theorem~\ref{thm:main}, states that 
for every countable class~$\mcl{L} \subseteq \powerset{\Sigma^*}$ of languages 
over a finite alphabet $\Sigma$, and for every alphabet $\Gamma$ with more than two symbols,
there exists a bijective encoding~$f : \Sigma^* \to \Gamma^*$ such that 
for every language $L \in \mcl{L}$ 
its image~$f[L]$ under $f$ is regular.

Furthermore we have shown that this result has a number of noteworthy consequences 
in language theory for regularity preserving functions, 
in computability theory for a concept for comparing the power of models of computation, 
and in the theory of automatic sequences 
for a generalization of this concept with respect to arbitrary number representations: 
\begin{enumerate}[(A)]
  \item
    There exists a computable bijective function $f : \Sigma^* \to \Gamma^*$
    such that the image function $f[\_]$ of $f$ is regularity preserving, 
    but the preimage function $f^{-1}[\_]$ is not
    (Corollary~\ref{cor:puzzle}).
    \smallskip

  \item 
    In the sense of \cite{boke:ders:2006},
    finite-state automata are as powerful
    as any countable decision model (e.g., Turing-ma\-chine deciders) 
    (Proposition~\ref{prop:moc:inj}).
    This even holds for the strongest notion of comparison in~\cite{boke:ders:2006},
    namely that with respect to bijective encodings 
    (Corollary~\ref{cor:moc:bij}).
    Similar counterintuitive consequences also affect computational models beyond decision models.
    \smallskip
    
  \item 
    For every countable class~$\mcl{C}$ of sets of natural numbers 
    there is a bijective number representation $r : \nat \to \Sigma^*$
    such that every set $S \in \mcl{C}$ 
    is \rrecognizable{} (i.e., there is a finite automaton deciding membership $n \in S$ on the input of $r(n)$)
    (Corollary~\ref{cor:main:recog}).

    More generally,
    for every countable class $\mcl{S} \subseteq \str{\Delta}$ of infinite sequences over a finite alphabet~$\Delta$,
    there exists a bijective number representation~$d : \nat \to \Gamma^*$
    such that every $\sigma \in \mcl{S}$ is $d$-automatic
    (Corollary~\ref{cor:main:auto}).

\end{enumerate}

These results also answer the questions in Section~\ref{sec:intro}
concerning the hierarchy of number representations.
From~(A) it follows that the hierarchy is proper: 
there are bijective representations $r_1$ and $r_2$
such that $r_1$ subsumes $r_2$, but not vice versa.
From (C) it follows that
%
%
every countable class $\mathcal{C} \subseteq \powerset{\nat}$ of languages
is contained in the countable class of
all $r$\nb-rec\-og\-niz\-able languages, for some representation $r$.
Moreover, (C) implies that
every injective number representation is subsumed by a bijective number representation,
and that no representation subsumes all others, 
since the class of \rrecognizable{} sets of natural numbers is always countable.

We conclude with two questions:
\begin{itemize}
  \item
    How far can computable bijective encodings $f : \Sigma^* \to \Gamma^*$ extend 
    the class of recognizable languages,
    that is,
    what classes
    $\mcl{L}_f = \{ L \subseteq \Sigma^* \where \text{$f[L]$ is a regular language} \}$    
    can we obtain for a computable bijective $f$?
    For example, is there a computable (bijective) encoding that makes precisely all context-free languages recognizable?
  \item 
    Rigo~\cite{rigo:2000} describes a class of number representations
    that 
         characterizes the 
                           morphic sequences.
    Our results entail the existence of a bijective 
                                                    representation $r : \nat \to \Sigma^*$
    such that every morphic sequence is $r$-auto\-ma\-tic.
    Is there a computable bijective representation $r$ such that \emph{precisely} the morphic sequences
    are $r$\nb-automatic?
\end{itemize}

\section*{Acknowledgment}
\addcontentsline{toc}{section}{Acknowledgment} 

We want to thank Nachum Dershowitz and Udi Boker 
for their remarks and a discussion about our results in Section~\ref{sec:moc},
as well as for several pointers to specific parts of their papers.

\bibliography{main.bib}

\begin{thebibliography}{10}

\bibitem{allo:shal:2003}
J.-P. Allouche and J.~Shallit.
\newblock {\em Automatic Sequences: Theory, Applications, Generalizations}.
\newblock Cambridge University Press, New York, 2003.

\bibitem{bers:boas:cart:peta:pin:2006}
J.~Berstel, L.~Boasson, O.~Carton, B.~Petazzoni, and J.-\'{E}. Pin.
\newblock Operations preserving regular languages.
\newblock {\em Theoretical Computer Science}, 354(3):405--420, 2006.

\bibitem{boke:2004}
U.~Boker.
\newblock {Comparing Computational Power}.
\newblock Master's thesis, Tel Aviv University, 2004.

\bibitem{boke:2008}
U.~Boker.
\newblock {\em {The Influence of Domain Interpretations on Computational
  Models}}.
\newblock PhD thesis, Tel Aviv University, 2008.

\bibitem{boke:ders:2006}
U.~Boker and N.~Dershowitz.
\newblock {Comparing Computational Power}.
\newblock {\em Logic Journal of the {IGPL}}, 14(5):633--647, 2006.

\bibitem{endr:grab:hend:2013}
J.~Endrullis, C.~Grabmayer, and D.~Hendriks.
\newblock {Mix-Automatic Sequences}.
\newblock In {\em Proc. of the 7th International Conference on Language and
  Automata Theory and Applications (LATA~2013)}, number 7810 in LNCS, 2013.

\bibitem{enge:hoog:2001}
J.~Engelfriet and H.~J. Hoogeboom.
\newblock {MSO} definable string transductions and two-way finite-state
  transducers.
\newblock {\em Transactions of the American Mathematical Society},
  2(2):216--254, 2001.

\bibitem{hopc:motw:ullm:2000}
J.~E. Hopcroft, R.~Motwani, and J.~D. Ullman.
\newblock {\em Introduction to Automata Theory, Languages and Computation}.
\newblock Addison-Wesley, 2nd edition, 2000.

\bibitem{kosa:1970}
S.~R. Kosaraju.
\newblock {Finite State Automata with Markers}.
\newblock In {\em Proc. 4th Annual Princeton Conference on Information Sciences
  and Systems}. Princeton, 1970.

\bibitem{kosa:1974}
S.~R. Kosaraju.
\newblock Regularity preserving functions.
\newblock {\em SIGACT News}, 6(2):16--17, 1974.

\bibitem{koze:1996}
D.~Kozen.
\newblock On regularity-preserving functions.
\newblock {\em Bulletin of the European Association for Theoretical Computer
  Science}, pages 131--138, 1996.

\bibitem{kvet:koub:1999}
P.~Kv\v{e}to\v{n} and V.~Koubek.
\newblock Functions preserving classes of languages.
\newblock In {\em Proc. Conf. on Developments in Language Theory}, pages
  81--102. World Scientific, 1999.

\bibitem{mato}
A.~B. Matos.
\newblock Regularity-preserving letter selections.
\newblock DCC-FCUP Interal Report.

\bibitem{pin:2009}
J.-\'{E}. Pin.
\newblock {Profinite Methods in Automata Theory}.
\newblock In {\em Proc. of the 26th Symposium on Theoretical Aspects of
  Computer Science (STACS~2009)}, pages 31--50. {IBFI Schloss Dagstuhl}, 2009.

\bibitem{pin:saka:1982}
J.-\'{E}. Pin and J.~Sakarovitch.
\newblock Some operations and transductions that preserve rationality.
\newblock {\em Theoretical Computer Science}, 145:277--288, 1982.

\bibitem{pin:silv:2005}
J.-\'{E}. Pin and P.~V. Silva.
\newblock A topological approach to transductions.
\newblock {\em Theoretical Computer Science}, 340(2):443--456, 2005.

\bibitem{rabi:1960}
M.~O. Rabin.
\newblock Computable algebra, general theory and theory of computable fields.
\newblock {\em Transactions of the American Mathematical Society},
  95(2):341--360, 1960.

\bibitem{rice:1956}
H.~G. Rice.
\newblock Recursive and recursively enumerable orders.
\newblock {\em Transactions of the American Mathematical Society},
  83(2):277--300, 1956.

\bibitem{rigo:2000}
M.~Rigo.
\newblock Generalization of automatic sequences for numeration systems on a
  regular language.
\newblock {\em Theoretical Computer Science}, 244(1-2):271--281, 2000.

\bibitem{roge:1967}
H.~Rogers.
\newblock {\em Theory of Recursive Functions and Effective Computability}.
\newblock MacGraw--Hill, 1967.

\bibitem{saka:2009}
J.~Sakarovitch.
\newblock {\em Elements of Automata Theory}.
\newblock Cambridge University Press, 2009.

\bibitem{seif:1974}
J.~I. Seiferas.
\newblock A note on prefixes of regular languages.
\newblock {\em SIGACT News}, 6(1):25--29, 1974.

\bibitem{seif:mcna:1976}
J.~I. Seiferas and R.~McNaughton.
\newblock Regularity-preserving relations.
\newblock {\em Theoretical Computer Science}, 2(2):147--154, 1976.

\bibitem{shal:1989}
J.~Shallit.
\newblock {A Generalization of Automatic Sequences}.
\newblock In {\em 6th Symposium on Theoretical Aspects of Computer Science
  (STACS~1989)}, volume 349 of {\em LNCS}, pages 156--167. Springer, 1989.

\bibitem{shap:1982}
S.~Shapiro.
\newblock Acceptable notation.
\newblock {\em Notre Dame Journal of Formal Logic}, 23(1):14--20, 1982.

\bibitem{stea:hart:1963}
R.~E. Stearns and J.~Hartmanis.
\newblock Regularity preserving modifications of regular expressions.
\newblock {\em Information and Control}, 6(1):55--69, 1963.

\end{thebibliography}

\end{document}